\documentclass[prl,twocolumn,amsmath,amssymb,10pt,aps,longbibliography,superscriptaddress,citeautoscript,bibnotes,floatfix,nobibnotes]{revtex4-2}

\usepackage{graphicx} 
\usepackage{dcolumn}
\usepackage{bm}
\usepackage{graphicx,color,xcolor}
\usepackage{physics}
\usepackage{comment}
\graphicspath{ {./Figures/} }
\usepackage{textcomp}
\usepackage[utf8]{inputenc}
\usepackage[T1]{fontenc}
\usepackage{amsmath}
\usepackage{amssymb}
\usepackage{amsthm}
\usepackage{graphicx}
\usepackage{subfigure}
\usepackage{tabularx}
\usepackage{booktabs}
\usepackage{multirow}
\usepackage{ragged2e}
\usepackage[colorlinks=true,linkcolor=blue,citecolor=blue]{hyperref} 

\usepackage{xcolor}

\newtheorem{theorem}{Theorem}

\newtheorem{lemma}[theorem]{Lemma}

\usepackage{soul}

\begin{document}

\title{Gluing Randomness via Entanglement: Tight Bound from Second Rényi Entropy}

\author{Wonjun Lee}
\email{wonjunphysics@gmail.com}
\affiliation{College of Natural Sciences, Korea Advanced Institute of Science and Technology, Daejeon, 34141, Korea}
\author{Hyukjoon Kwon}
\email{hjkwon@kias.re.kr}
\affiliation{School of Computational Sciences, Korea Institute for Advanced Study, Seoul, 02455, South Korea}
\author{Gil Young Cho}
\email{gilyoungcho@kaist.ac.kr}
\affiliation{Department of Physics, Korea Advanced Institute of Science and Technology, Daejeon, 34141, Korea}
\affiliation{Center for Artificial Low Dimensional Electronic Systems, Institute for Basic Science, Pohang, 37673, Republic of Korea}

\begin{abstract}
    The efficient generation of random quantum states is a long-standing challenge, motivated by their diverse applications in quantum information processing tasks. In this manuscript, we identify entanglement as the key resource that enables local random unitaries to generate global random states by effectively gluing randomness across the system. Specifically, we demonstrate that approximate random states can be produced from an entangled state $\ket{\psi}$ through the application of local random unitaries. We show that the resulting ensemble forms an approximate state design with an error saturating as $\Theta(e^{-\mathcal{N}_2(\psi)})$, where $\mathcal{N}_2(\psi)$ is the second R\'enyi entanglement entropy of $\ket{\psi}$. Furthermore, we prove that this tight bound also applies to the second R\'enyi entropy of coherence when the ensemble is constructed using coherence-free operations. These results imply that, when restricted to resource-free gates, the quality of the generated random states is determined entirely by the resource content of the initial state. Notably, we find that among all $\alpha$-R\'enyi entropeis, the second R\'enyi entropy yields the tightest bounds. Consequently, these second R\'enyi entropies can be interpreted as the maximal capacities for generating randomness using resource-free operations. Finally, moving beyond approximate state designs, we utilize this entanglement-assisted gluing mechanism to present a novel method for generating pseudorandom states in multipartite systems from a locally entangled state via pseudorandom unitaries in each of parties. 
\end{abstract}

\maketitle

\noindent
\large{{\bf Introduction}}\\
Quantum resource theories provide a rigorous framework for identifying the key ingredients of quantum advantage~\cite{jozsa2003,Brandao2015,Chitambar2019}. Among these, entanglement, magic, and coherence emerge as fundamental pillars~\cite{Horodecki2009,Bravyi2005,Streltsov2017}. This non-interchangeable family is essential for achieving universal quantum computation. A critical benchmark for this capability is the generation of Haar random states, which are the most generic and complex states in Hilbert space. Since producing these states requires the construction of arbitrary unitary operators, the ability to do so is fundamentally limited by the available budget of entanglement, magic, and coherence~\cite{Page1993,Hayden2006Aspects,Napoli2016,Leone2022}. Consequently, the amount of these resources directly constrains the degree of randomness a quantum processor can ultimately produce.

While random states are useful for various quantum information processing tasks~\cite{Ji2018,Knill_2008,Boixo_2018,Huang_2020,elben2023} as well as understanding complex quantum many-body systems~\cite{Srednicki1994, Piroli_2020, Matthew2023}, their exact implementation is challenging due to prohibitive resource requirements~\cite{knill1995}. To address this, approximation schemes such as approximate designs and pseudorandom states have been developed~\cite{Andris2007,Ji2018}. Based on these schemes, it has been shown that high-quality approximations can have significantly smaller amounts of  resources~\cite{aaronson2023quantum,gu2023little,lee2025shallow}. Recently, efficient methods for generating these approximated states have been substantially developed~\cite{Brand_o_2016,Haferkamp2022randomquantum,Ho_2022,Harrow2023,Cotler2023,haah2024,lee2025shallow,cui2025,zhang2025}. A particularly surprising result is the gluing lemma, which demonstrates that small random unitary operators with overlapping supports can be combined together to form a larger approximate random unitary~\cite{schuster2024}. This allows for the generation of a global approximate unitary using only two layers of local random unitaries. While effective, this method relies on intermediate unitaries to connect unitaries in disjoint regions. Therefore, to realize approximate random states using this method only with entanglement-free local operations, one should share different random states among multiple parties for every realization. 

In this manuscript, we demonstrate that approximate random states can be generated in multi-party systems using a fixed entangled state $\ket{\psi}$ and a single layer of local random unitaries. Importantly, $\ket{\psi}$ can be {an arbitrary entangled state, enabling to use optimized states for specific hardware platforms such as Bell pairs.} The key that allows this is that the entanglement within $\ket{\psi}$ effectively connects the local unitaries of individual parties. This establishes the role of entanglement in the gluing lemma introduced in Ref.~\cite{schuster2024} and provide its generalization to random quantum states, facilitating their more efficient generation. We find that the quality of the connection across each partition to produce approximate state design is precisely quantified by the second R\'enyi entanglement entropy $\mathcal{N}_2(\psi)$ of $\ket{\psi}$. From this, we show that $\mathcal{N}_2(\psi)$ yields the tightest bound on the error of the best approximate state design achievable from $\ket{\psi}$ without the use of entangling gates. Since our protocol only uses local unitaries, the tightness of this bound implies interestingly that classical communication, which cannot increase $\mathcal{N}_2(\psi)$, offers only marginal benefits for generating approximate random states.

Furthermore, by analogy to entanglement, we show that the second R\'enyi entropy of coherence quantifies the maximal achievable randomness without using coherent gates. A similar bound for stabilizer orbits is established in Ref.~\cite{bittel2025}. We emphasize that while exact characterization of these resources is often intractable~\cite{GURVITS2004, Baumgratz2014, Howard2017, aaronson2023quantum}, estimating their R\'enyi entropies is far more accessible in practice~\cite{Hastings2010,vanEnk2012,Huang_2020,Leone2022,oliviero2022,Haug2024}. On the other hand, practical uses of any single R\'enyi entropy are frequently obscure. Our findings provide a clear operational meaning of the second R\'enyi entropy as maximal capacities for generating randomness using resource-free operations.

\begin{figure}[t]
    \includegraphics[width=\linewidth]{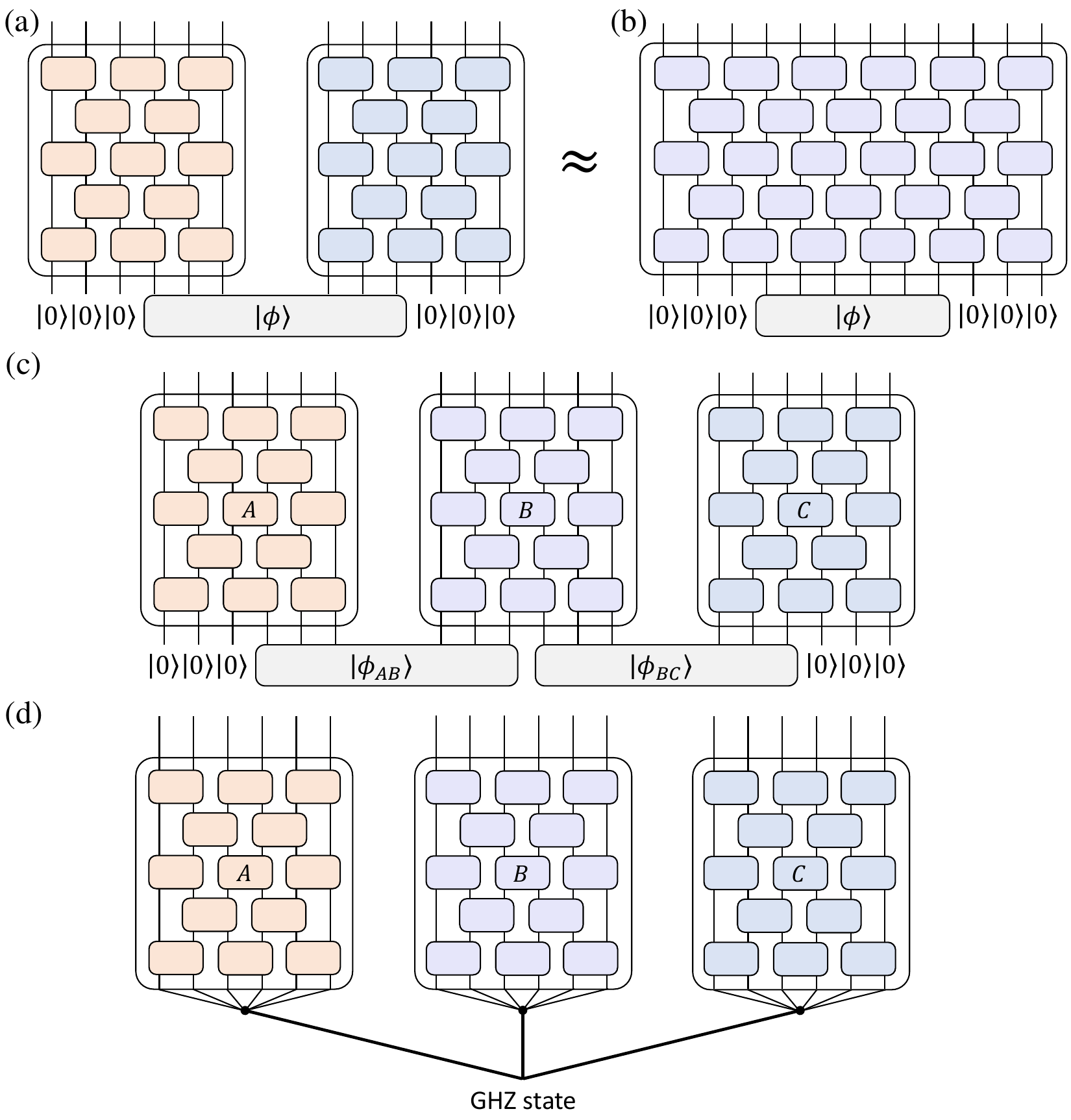}
    \centering
    \caption{\textbf{Approximate designs from entangled states and local operations.} (a) Alice (left) and Bob (right) initially share a state $\ket{\phi}$ having a non-vanishing second R\'enyi entanglement entropy. They apply local random unitary operators to scramble their systems. (b) Global random states are generated by random circuits. States generated in (a) global approximate random states in (b). (c,d) Alice (left) and Bob (mid) as well as Bob and Charlie (right) initially share states $\ket{\phi_{AB}}$ and $\ket{\phi_{BC}}$ or the GHZ state. These initial states have non-trivial second R\'enyi entanglement entropies. They then apply local random unitary operators and generate approximate random states.}
    \label{fig:example}
\end{figure}

\vspace{1em}
\noindent{\large{\bf Results}}\\
\noindent{{\bf Gluing local randomness via entangled state}}\\
We first show that approximate state designs can be generated from an entangled state $\ket{\psi}$ and local operations for bipartite systems. This result is illustrated in Fig.~\ref{fig:example}(a,b). We next generalize this to multipartite systems with local operations as well. As demonstrations, we consider methods for generating approximate state designs from quantum Markov chains~\cite{Hayden2004} and higher-level GHZ states~\cite{Ryu2013} as illustrated in Fig.~\ref{fig:example}(c,d). Errors of approximate designs are given by the second R\'enyi entanglement entropies of the states.

Before providing details, we introduce the constant entanglement orbit $\mathcal{E}_{\mathfrak{C},\mathrm{ent.}}(\psi)$ of a state $\ket{\psi}$ with a set $\mathfrak{C}$ of disjointed regions as the ensembles constructed by applying Haar random unitary operators supported in each of regions of $\mathfrak{C}$ to $\ket{\psi}$. The formal definition of this ensemble is given in the Methods section, in parallel with the definitions of approximate state designs and the second Rényi entropy. For bipartite systems, we drop the subscript `$\mathfrak{C}$'. Below, we first show that $\mathcal{E}_{\mathrm{ent.}}(\psi)$ forms an approximate state design for bipartite systems with error upper bounded by the second R\'enyi entanglement entropy $\mathcal{N}_2(\psi)$.

\begin{theorem}\label{thm:upper-bound}
    For any pure state $\ket{\psi}$ and an integer $t\geq 2$, $\mathcal{E}_{\mathrm{ent.}}(\psi)$ forms approximate state $t$-designs with error $O(t^2 e^{-\mathcal{N}_2(\psi)})$.
\end{theorem}
We defer the detailed proof to Supplementary Note 1 and instead present several key insights underlying this theorem. First, averaging over Haar random unitaries approximately gives permutation operators that mix $t$-copies~\cite{schuster2024}. Second, permutation operators in different regions are synchronized by entanglement of $\ket{\psi}$. The quality of this synchronization is precisely given by $O(t^2\exp(-\mathcal{N}_2(\psi)))$, giving the stated upper bound on the approximation error. As we show below, these insights can be used to generalize the theorem to multipartite settings.

\textbf{Theorem}~\ref{thm:upper-bound} immediately implies that the error of approximate random states can be suppressed exponentially by increasing the second R\'enyi entanglement entropy of the input state $\ket{\psi}$. More precisely, starting from a simple resource state, such as products of a few Bell pairs and computational states, we can generate approximate random states by applying entanglement-free operations as illustrated in Fig.~\ref{fig:example}(a).

Next, we generate approximate state designs {in multipartite systems} using quantum Markov chains~\cite{Hayden2004}. Specifically, we consider $\mathcal{E}_{\mathfrak{C},\mathrm{ent.}}(\psi)$ constructed from a pure quantum Markov chain $\ket{\psi}_\mathfrak{C}$ with a sequence $\mathfrak{C}$ of disjointed regions $\{A_i\}_{i=1}^l$ by applying Haar random unitary operators supported in $\{A_i\}_{i=1}^l$ (see Fig.~\ref{fig:example}(c) for $\abs{\mathfrak{C}}=3$).
\begin{theorem}\label{thm:markov-upper-bound}
    For any pure quantum Markov chain $\ket{\psi}_{\mathfrak{C}}$ with a sequence $\mathfrak{C}=\{A_i\}_{i=1}^l$, $\mathcal{E}_{\mathfrak{C},\mathrm{ent.}}(\psi)$ forms an approximate state $t$-design with error $O\left(t^2\sum_{i=1}^{l-1} e^{-\mathcal{N}_2^{(\mathcal{A}_i)}(\psi)}\right)$ and $\mathcal{A}_i=\bigcup_{j=1}^{i}A_j$.
\end{theorem}
We prove this theorem in Supplementary Note 2. The above theorem admits a particularly simple interpretation when applied to a tensor product of initially distributed Bell pairs between neighboring regions, which constitutes a pure quantum Markov chain. In this case, similar to the bipartite case \textbf{Theorem}~\ref{thm:upper-bound}, approximate state designs can be generated by initially sharing Bell pairs across each partition and applying local random unitaries. The initial Bell pairs, more generally, second R\'enyi entanglements of the partitions of the Markov chain, effectively glue these local random unitaries and generate global approximate state designs, providing a state version of the gluing lemma introduced in Ref.~\cite{schuster2024}. We compare our findings with the original gluing lemma for generating approximate unitary designs in Fig.~\ref{fig:gluing}.

For the same multipartition $\{A_i\}_{i=1}^l$, approximate state designs can also be generated from a genuinely multipartite entangled state. As a demonstration, we consider higher-level GHZ states as the following theorem. We prove this theorem in Supplementary Note 3.
\begin{theorem}
    For any set $\mathfrak{C}$ of disjointed regions $\{A_i\}_{i=1}^l$ and a positive integer $d\leq \min_{1\leq i\leq l}(2^{|A_i|})$, $\mathcal{E}_{\mathfrak{C},\mathrm{ent.}}(\psi)$ generated from the $d$-level GHZ state $\ket{\psi}=\frac{1}{\sqrt{d}}\sum_{x=1}^d\bigotimes_{i=1}^l\ket{x}_{A_{i}}$ forms an approximate state $t$-design with error $O\left(t^2/d\right)$.
\end{theorem}

\begin{figure}[t]
    \includegraphics[width=\linewidth]{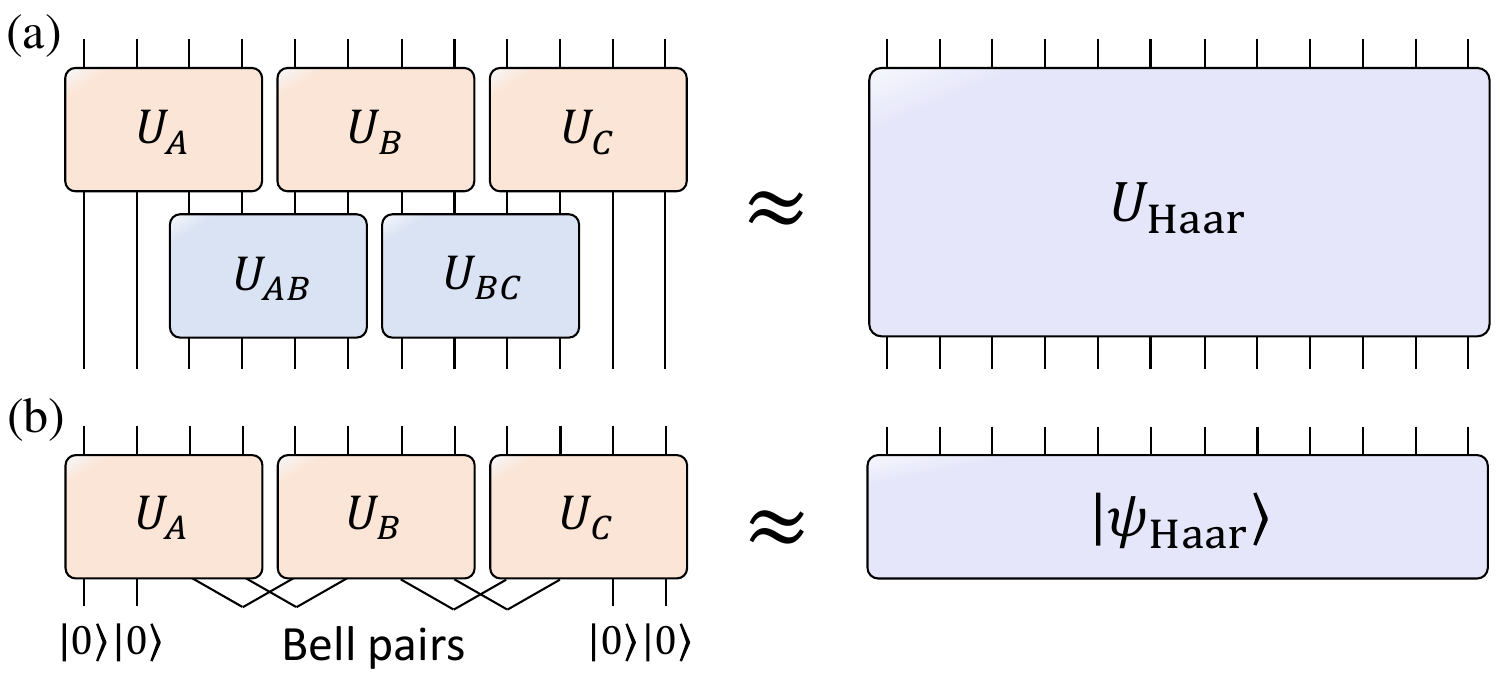}
    \centering
    \caption{\textbf{Gluing local randomness.} (a) Small approximate unitary designs $U_{AB}$ and $U_{BC}$ glue other approximate unitary designs $U_A$, $U_B$, and $U_C$ to give a global approximate unitary design. (B) Bell pairs glue local approximate unitary designs $U_A$, $U_B$, and $U_C$ and enable the generation of a global approximate state design by a single layer of local random unitaries.}
    \label{fig:gluing}
\end{figure}

\begin{figure}[t]
    \includegraphics[width=\linewidth]{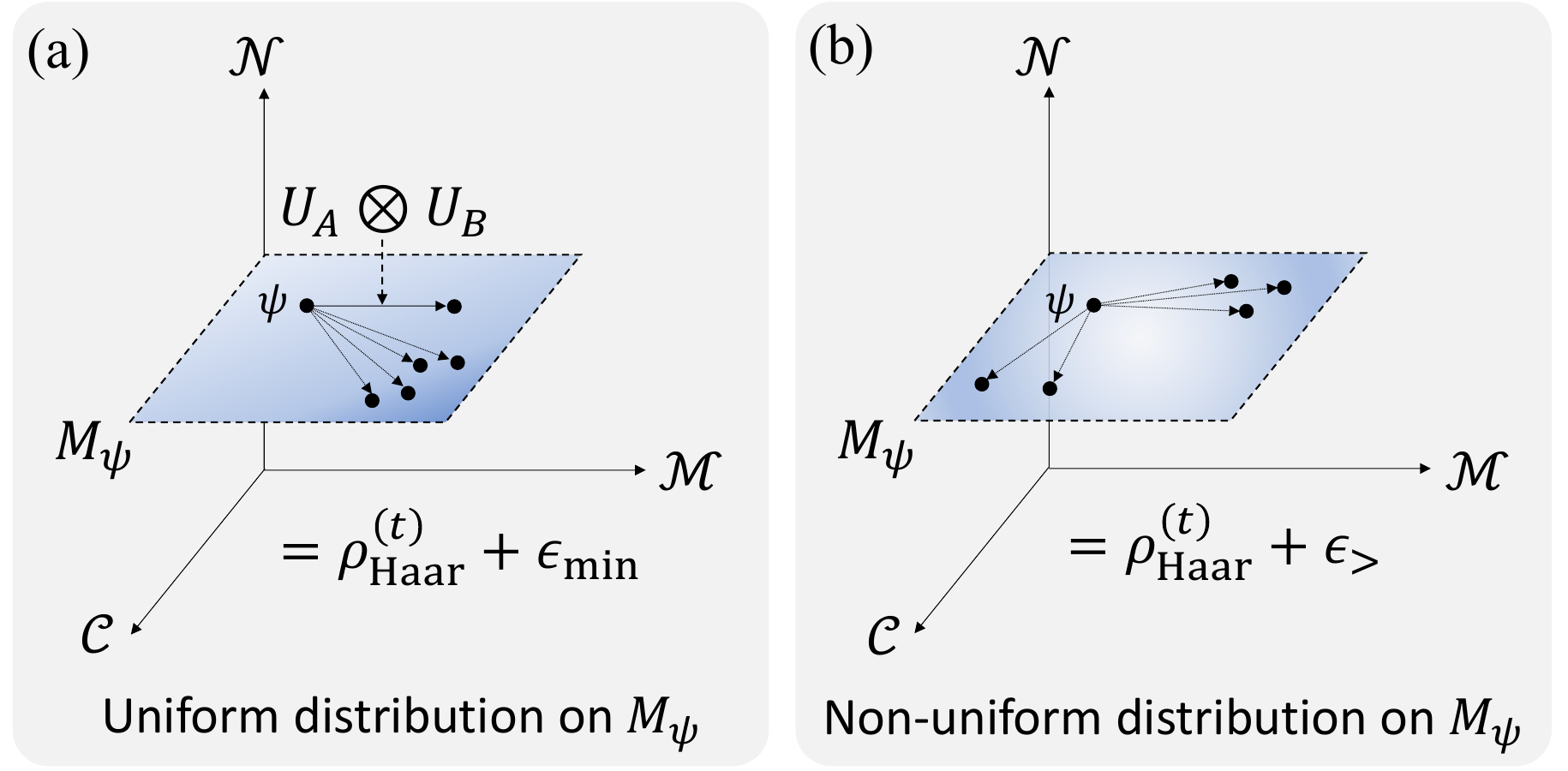}
    \centering
    \caption{\textbf{Quantifying maximal achievable randomness.} (a,b) $\mathcal{N}$, $\mathcal{M}$, and $\mathcal{C}$ represent amounts of entanglement, magic, and coherence in arbitrary fixed measures. $M_\psi$ is the set of states that can be constructed by applying entanglement-free operators $U_A\otimes U_B$ to $\ket{\psi}$. $\rho^{(t)}_\mathrm{Haar}$ is the $t$-th moments of the ensemble of Haar random states. (a) The ensemble $\mathcal{E}_\mathrm{ent.}(\psi)$ of states in $M_\psi$ with the uniform distribution according to the Haar measure has the minimum deviation $\epsilon_\mathrm{min}$ from $\rho^{(t)}_\mathrm{Haar}$. States in this ensemble typically have near maximal magic and coherence. (b) The ensemble of states in $M_\psi$ with a non-uniform distribution has a larger deviation $\epsilon_>$ from $\rho^{(t)}_\mathrm{Haar}$, \textit{i.e.}, $\|\epsilon_>\|_1\geq\|\epsilon_\mathrm{min}\|_1$. 
    }
    \label{fig:orbit}
\end{figure}

\vspace{1em}
\noindent{\bf Second R\'enyi entropy as maximal achievable randomness}\\
We previously established that an approximate state design can be generated from a fixed state $\ket{\psi}$, with an error upper bounded by the second R\'enyi entanglement entropy of $\ket{\psi}$. Here, we establish that this upper bound is tight. Moreover, we find that an analogous result holds for coherence. Furthermore, we propose that the second R\'enyi entropies for entanglement and coherence can be used to quantify the maximal achievable Haar randomness in the context of approximate state designs, as illustrated in Fig.~\ref{fig:orbit}. 

We first consider coherence and its relation to the maximal achievable Haar randomness. To make this statement precise, we introduce the constant coherence orbit
$\mathcal{E}_\mathrm{coh.}(\psi)$ of a state $\ket{\psi}$ defined as the ensemble generated by applying coherence-free unitary operators to $\ket{\psi}$ (see Methods for a precise definition). We show that $\mathcal{E}_\mathrm{coh.}(\psi)$ forms an approximate state design, with an error that is exponentially suppressed by the second R\'enyi entropy of coherence, $\mathcal{C}_2(\psi)$, in the following lemma. The proof is in Supplementary Note 4.

\begin{lemma}\label{thm:coh-upper-bound}
    For any pure state $\ket{\psi}$ and an integer $t\geq 2$, $\mathcal{E}_{\mathrm{coh.}}(\psi)$ forms an approximate state $t$-design with error $O(t^2 e^{-\mathcal{C}_2(\psi)})$.
\end{lemma}

We also find that the upper bounds of errors in \textbf{Theorem}~\ref{thm:upper-bound} and \textbf{Lemma}~\ref{thm:coh-upper-bound} are in fact \textit{tight} due to the following theorem. 

\begin{theorem}\label{thm:tight-bound}
    For any pure state $\ket{\psi}$ and a constant $t$, $\mathcal{E}_\mathrm{ent.}(\psi)$ and $\mathcal{E}_\mathrm{coh.}(\psi)$ form approximate state $t$-designs with error $\Theta(e^{-\mathcal{N}_2(\psi)})$ and $\Theta(e^{-\mathcal{C}_2(\psi)})$, respectively.
\end{theorem}
\begin{proof}
    Due to Ref.~\cite{lee2025shallow}, if an ensemble of states forms an $\epsilon$-approximate $t$-design for some constant $t$, then its mean $\mathcal{N}_2$ and $\mathcal{C}_2$ are lower bounded by $\log\left( \Omega(1/\epsilon)\right)$. This, together with \textbf{Theorem}~\ref{thm:upper-bound} and \textbf{Lemma}~\ref{thm:coh-upper-bound}, gives the stated tight bounds.
\end{proof}
A similar discussion regarding magic appears in Ref.~\cite{bittel2025}, where the stabilizer orbit of a state $\ket{\psi}$ forms an approximate state $t$-designs. The error is found to be upper bounded by $O(2^{t^2/2}e^{-\mathcal{M}_2(\psi)})$ with the second stabilizer R\'enyi entropy $\mathcal{M}_2(\psi)$ of $\ket{\psi}$~\cite{Leone2022} (The definition of $\mathcal{M}_2(\psi)$ can be found in Methods.). However, the tightness of this upper bound has not been proven. In contrast, we establish the tight bounds for entanglement and coherence.

So far, we have shown that the approximation errors are tightly bounded by the second Rényi entropies of entanglement and coherence. We further demonstrate that these errors are minimal among all ensembles constructible using resource-free operations, as established by the following theorem. The proof can be found in Supplementary Note 5.
\begin{theorem}\label{thm:uniform-optimal}
    For any pure state $\ket{\psi}$, $\mathcal{E}_\mathrm{ent.}(\psi)$ and $\mathcal{E}_\mathrm{coh.}(\psi)$ form approximate state designs with minimum errors among those can be constructed using entanglement and coherence-free gates, respectively.
\end{theorem}

This constitutes one of our central results. The theorem establishes a fundamental limit: any approximation of randomness generated using only resource-free operations is bounded by the second R\'enyi entropies of the initial state. This result is illustrated in Fig.~\ref{fig:orbit}.

We next present another key result. While our analysis has mainly focused on the second R\'enyi entropies, one may ask whether R\'enyi entropies of other orders $\alpha$ may provide tighter approximation errors. However, we find that among all Rényi entropies, the second Rényi entropy yields the tightest bounds.

\begin{theorem}[Informal]
    The second R\'enyi entropies give the tightest bounds for the approximation errors.
\end{theorem}
We prove this theorem in Supplementary Note 6. Therefore, the second R\'enyi entropy can be interpreted as the maximal achievable Haar randomness without using resourceful gates.

Finally, we present the following observation: an ensemble constructed from the constant-entanglement–and–constant-coherence orbits of an input state $\ket{\psi}$ forms an approximate state design with a nearly minimal approximation error. To better appreciate this statement, consider the ensemble $\mathcal{E}_{\rm ent.}(\psi)$. Although $\mathcal{E}_{\rm ent.}(\psi)$ consists of states with fixed entanglement, their coherence is typically close to maximal. This follows directly from the construction of $\mathcal{E}_{\rm ent.}(\psi)$, which is generated by applying Haar-random unitary operators independently within each region to $\ket{\psi}$. One might therefore expect that at least one of these resources must be extensively consumed to generate high-quality approximate random states. Surprisingly, this is not the case. To formalize this observation, we introduce in Methods the constant-entanglement–and–constant-coherence orbit $\mathcal{E}_{\mathrm{e.c.}}(\psi)$ of an input state $\ket{\psi}$, which will be used in the theorem below.

\begin{theorem}\label{thm:rsp}
    For a constant $t\geq 2$, there is a quantum state $\ket{\psi}$ with $\mathcal{N}_2(\psi)=\mathcal{C}_2(\psi)$ which gives and the constant-entanglement–and–constant-coherence orbit $\mathcal{E}_\mathrm{e.c.}(\psi)$ forming an approximate state $t$-design with an error that scales with $\Theta(e^{-\mathcal{N}_2(\psi)})$. 
\end{theorem}

We can construct an ensemble of states satisfying this theorem by slightly modifying the ensemble $\mathcal{E}_\mathrm{rsp}$ of random subset phase states~\cite{aaronson2023quantum}. $\mathcal{E}_\mathrm{rsp}$ is defined by two random objects: random permutation $p:\{0,1\}^n\rightarrow\{0,1\}^n$ and random function $f:\{0,1\}^n\rightarrow\{0,1\}$ with the number of qubits $n$. Each state in $\mathcal{E}_\mathrm{rsp}$ with the subsystem size $k$ is given by
\begin{equation*}
    \ket{\psi_{p,f}} = \frac{1}{\sqrt{2^k}}\sum_{b\in\{0,1\}^k}(-1)^{f(b,0^{n-k})}\ket{p(b,0^{n-k})}.
\end{equation*}
By construction, these states have constant coherence of $\mathcal{C}_2(\psi)=k\log 2$. Furthermore, $\mathcal{E}_\mathrm{rsp}$ is known to form an $\epsilon$-approximate state $t$-design with $\epsilon=\Theta(2^{-k})$~\cite{lee2025shallow}. In Supplementary Note 7, we show that {this modified ensemble has constant entanglement of $\mathcal{N}_2(\psi)=k\log 2$ while having the same order of approximation error $\Theta(2^{-k})=\Theta(e^{-\mathcal{N}_2(\psi)})=\Theta(e^{-\mathcal{C}_2(\psi)})$. }

\vspace{1em}
\noindent{\bf Generation of pseudorandom states}\\
In addition to approximate state designs, pseudorandom states serve as another widely accepted approximation scheme for Haar random states~\cite{Ji2018,brakerski2019,aaronson2023quantum,gu2023little}. These states are computationally indistinguishable from Haar random states, \textit{i.e.}, there is no poly-time quantum algorithm distinguishing them~\cite{Ji2018}. We find that our construction of gluing randomness (Fig.~\ref{fig:gluing}) can also be used to generate pseudorandom states, a result that can be directly inferred from the operational meaning of approximate state designs. Specifically, if an ensemble $\mathcal{E}$ of states forms an $\epsilon$-approximate state $t$-design, then any quantum algorithm using $t$ copies of states can distinguish $\mathcal{E}$ from Haar random states with a success probability at most $(1+\epsilon)/2$. Consequently, if $\epsilon$ is $1/\omega(\mathrm{poly}(n))$ for any $t=\mathrm{poly}(n)$ with the number of qubits $n$, then states uniformly sampled from $\mathcal{E}$ are pseudorandom. Based on this, we establish the following theorem.
\begin{theorem}
    For a $n$-qubit pure quantum Markov chain $\ket{\psi}_\mathfrak{C}$ with a sequence $\mathfrak{C}=\{A_i\}^l_{i=1}$ such that $|A_i|=\mathrm{polylog}(n)$ for all $i\in[1,l]$, if $\mathcal{N}_2^{(\mathcal{A}_i)}(\psi)$ is at least $\mathrm{polylog}(n)$ for all $i\in[1,l-1]$ with $\mathcal{A}_i=\bigcup_{j=1}^i A_j$, then pseudorandom states can be generatedf rom $\ket{\psi}_\mathfrak{C}$ by applying $\mathrm{polylog}(n)$ depth local operations.
\end{theorem}
\begin{proof}
    Let $\mathcal{E}_\mathrm{pseudo}$ be an ensemble generated by applying pseudorandom unitaries in each of $\{A_i\}_{i=1}^l$ to $\ket{\psi}_\mathfrak{C}$. Such a $\mathrm{polylog}(n)$-sized pseudorandom unitary can be implemented in $\mathrm{polylog}(n)$ depth~\cite{schuster2024,ma2025,schuster2025}. By the definition of pseudorandom unitary, $\mathcal{E}_\mathrm{pseudo}$ is computationally indistinguishable from $\mathcal{E}_{\mathfrak{C},\mathrm{ent.}}(\psi)$. Since $\mathcal{N}_2^{(\mathcal{A}_i)}(\psi)$ is $\omega(\log n)$ for all $i\in[1,l-1]$, due to \textbf{Theorem}~\ref{thm:markov-upper-bound}, $\mathcal{E}_{\mathfrak{C},\mathrm{ent.}}(\psi)$ forms an approximate state $t$-design with error upper bounded by $t^2/g(n)$ for some $g\in \omega(\mathrm{poly}(n))$. Therefore, for any $t=\mathrm{poly}(n)$, the error is given by $1/\omega(\mathrm{poly}(n))$. Thus, states in $\mathcal{E}_\mathrm{pseudo}$ are computationally indistinguishable from Haar random states.
\end{proof}
This implies that pseudorandom states can be generated by sharing $\mathrm{polylog}(n)$ Bell pairs and applying a single layer of small pseudorandom unitaries as illustrated in Fig.~\ref{fig:gluing}. To the best of our knowledge, this constitutes the first construction of pseudorandom states using locally shared Bell pairs in multipartite systems. We note that the required second R\'enyi entanglement entropy saturates the lower bound established for pseudorandom states in Ref.~\cite{aaronson2023quantum}.

\vspace{1em}
\noindent{\large{\bf Discussion}}\\
In this work, we demonstrate that approximate state designs can be generated from an entangled state with local random unitaries. Furthermore, we show that the approximation error is tightly bounded by the second R\'enyi entanglement entropy of the initial state, providing its clear operational meaning. We show that the approximation error decays exponentially as the second R\'enyi entanglement entropy increases, indicating that the initial number of Bell pairs required to achieve the desired error could be exponentially small. This finding also enables the generation of pseudorandom states~\cite{Ji2018} from fixed shared states such as Bell pairs and a single layer of local pseudorandom unitaries. Our approach leads to a more efficient generation of approximate random states without requiring the two-layer construction for gluing unitaries~\cite{schuster2024}, making it particularly well suited for distributed quantum computing scenarios~\cite{CALEFFI2024,Main2025}. At the same time, we find that the same tight bound holds for constantly coherent ensembles. This implies that the second R\'enyi entropies are not only more experimentally accessible measures of quantum resources, but also operationally useful quantifiers of a state's capacity to generate randomness.

Our findings may have profound connections to black hole physics, where Bell pairs generated at the event horizon undergo fast scrambling dynamics~\cite{Hayden_2007,Sekino_2008}. It would be interesting to generalize our findings to random unitaries under symmetry constraints such as energy-conserving chaotic Hamiltonian dynamics. Given our finding that approximate random states can be generated from certain multipartite entangled states, it would also be interesting to explore whether there exists a multipartite entangled quantum state that cannot be used to generate approximate random states via local operations. Further, while we find tight bounds for entanglement and coherence, it is still an open question if the same R\'enyi entropy bound holds for magic as well. 

\vspace{1em}
\noindent{\large{\bf Methods}}\\
\noindent{\bf Faithful quantum resources}\\
Entanglement, magic, and coherence are necessary and sufficient resources for universal computation. This can be understood by considering the universal gate set of $\{CNOT,T,H\}$. Each element increases entanglement, magic, and coherence, respectively, while preserving other resources. Therefore, these resources can be thought of as faithful quantum resources.

These resources can be quantified by their R\'enyi entropies. Let $\ket{\psi}$ be a $n$-qubit state. We set $A$ be a subsystem of the system. Then, the Schmidt decomposition of $\ket{\psi}$ with respect to $A$ is given by $\ket{\psi}=\sum_{x\in[N_A]}\lambda_x\ket{x}_A\otimes \ket{x}_{A^c}$ with $N_A=2^{\abs{A}}$. Then, the second R\'enyi entanglement entropy is given by~\cite{Pasquale2004}
\begin{equation*}
    \mathcal{N}^{(A)}_2(\psi) = -\log\left(\sum_{x\in[N_A]}\lambda_x^4\right).
\end{equation*}
For simplicity, we suppress the superscript $(A)$, retaining it only when necessary. Next, let us set expectation values of Pauli strings to be $w_x=N^{-1}\abs{\bra{\psi}P_x\ket{\psi}}^2$ for any $x\in[N^2]$ with $N=2^n$. Then, the second R\'enyi entropy for magic, known as the second stabilizer R\'enyi entropy, is given by~\cite{Leone2022}
\begin{equation*}
    \mathcal{M}_2(\psi) = -\log\left(\sum_{x\in[N^2]}w_x^2\right)-\log N.
\end{equation*}
Finally, let us expand $\ket{\psi}$ in the computational basis as $\ket{\psi}=\sum_{x\in[N]}c_x\ket{x}$. We define the second R\'enyi entropy for coherence as
\begin{equation*}
    \mathcal{C}_2(\psi) = -\log\left(\sum_{x\in[N]}\abs{c_x}^4\right).
\end{equation*}
This is vanishing for computational states, which verify the faithfulness, and invariant under coherence preserving unitary operators $\{U_p U_f|U_p\in P(\mathcal{H}), U_f\in [U(1)]^{\oplus N}\}$ with the sets of permutation operators $P(\mathcal{H})$ in the Hilbert space $\mathcal{H}$ of the system and phase factors $U(1)$. 

\vspace{1em}
\noindent{\bf Approximate state design}\\
The ensemble of Haar random states, denoted as $\mathcal{E}_\mathrm{Haar}$, is defined as the set of all quantum states distributed uniformly according to the Haar measure. Since $\mathcal{E}_\mathrm{Haar}$ contains arbitrary quantum states, its realization requires significant amounts of entanglement, magic, and coherence, which are resources necessary for universal quantum computation. Furthermore, it is well known that typical states in this ensemble possess these resources at near maximal levels, making their experimental implementation challenging. This difficulty motivates the development of approximation schemes for true random states, $\mathcal{E}_\mathrm{Haar}$. 

A prominent approximation scheme is the approximate state design. An ensemble $\mathcal{E}$ of states forms an $\epsilon$-approximate state $t$-design if the $t$-th moments of its states approximates that of $\mathcal{E}_\mathrm{Haar}$ up to an additive error $\epsilon$ in the trace distance:
\begin{equation*}
    \mathrm{TD}\left(\mathbb{E}_{\psi\sim\mathcal{E}}[\ketbra{\psi}^{\otimes t}],\mathbb{E}_{\phi\sim\mathcal{E}_\mathrm{Haar}}[\ketbra{\phi}^{\otimes t}]\right) \leq \epsilon,
\end{equation*}
where the trace distance between two states $\rho$ and $\sigma$ is defined as $\mathrm{TD}(\rho,\sigma)=\frac{1}{2}\|\rho-\sigma\|_1$. Operationally, this bound implies that the success probability of distinguishing $\mathcal{E}$ from $\mathcal{E}_\mathrm{Haar}$ via a $t$-copy measurement is at most $(1+\epsilon)/2$. Unlike the true random states, this statistically close ensemble can be efficiently implemented using shallow quantum circuits. Furthermore, their average resource contents can be significantly lower than the maximal values associated with Haar random states.

\vspace{1em}
\noindent{\bf Constant resource orbits}\\
Let $\mathfrak{C}$ be a set of disjointed regions $\{A_i\}_{i=1}^l$ that covers the entire system. Then, the constant entanglement orbit $\mathcal{E}_{\mathfrak{C},\mathrm{ent.}}(\psi)$ of a state $\ket{\psi}$ with $\mathfrak{C}$ is defined as
\begin{equation*}
    \mathcal{E}_{\mathfrak{C},\mathrm{ent.}}(\psi) = \left\{\bigotimes_{i=1}^l U_{A_i}\ket{\psi}|U_{A_i}\sim\mathcal{U}(\mathcal{H}_{A_i})\right\},
\end{equation*}
where $\mathcal{H}_{A_i}$ is the Hilbert space within $A_i$, and $\mathcal{U}(\mathcal{H}_{A_i})$ is the ensemble of Haar random unitary operators in $\mathcal{H}_{A_i}$.

Next, the constant coherence orbit $\mathcal{E}_\mathrm{coh.}(\psi)$ of $\ket{\psi}$ is defined as
\begin{equation*}
    \mathcal{E}_\mathrm{coh.}(\psi) = \left\{U \ket{\psi} | U \sim \mathcal{PF}(\mathcal{H}) \right\}.
\end{equation*}
Here, $\mathcal{PF}(\mathcal{H})$ is the ensemble of products of unitary operators $U_p$ and $U_f$, where $U_p$ is uniformly sampled from the set of permutations $P(\mathcal{H})$, and $U_f$ is a diagonal unitary operator acting on $\mathcal{H}$ whose elements are $e^{i\theta}$ with $\theta$ uniformly sampled from $[0,2\pi)$.

Finally, we define the constant entanglement and coherence orbit $\mathcal{E}_\mathrm{e.c.}(\psi)$ of $\ket{\psi}$ with a bipartition $\{A,A^c\}$ as
\begin{equation*}
    \begin{split}
        \mathcal{E}_\mathrm{e.c.}(\psi) 
        &= \{U_A\otimes U_{A^c}\ket{\psi}|U_A\sim\mathcal{PF}(\mathcal{H}_A),\\
        &\qquad\qquad\qquad\qquad U_{A^c}\sim\mathcal{PF}(\mathcal{H}_{A^c})\}.        
    \end{split}
\end{equation*}

\vspace{1em}
\noindent {\bf References}

\bibliography{Ref}

\vspace{1em}
\noindent{\bf Acknowledgments}\\
W.L. is supported by the KAIST Jang Young Sil Fellow Program. H.K. is supported by the KIAS Individual Grant No. CG085302 at Korea Institute for Advanced Study and National Research Foundation of Korea (Grants No.~RS-2023-NR119931, No.~RS-2024-00413957, and No.~RS-2024-00438415) funded by the Korea Government (MSIT). G.Y.C. is financially supported by Samsung Science and Technology Foundation under Project Number SSTF-BA2401-03, the NRF of Korea (Grants No. RS-2023-00208291, RS-2024-00410027, RS-2023-NR119931, RS-2024-00444725, RS-2023-00256050, IRS-2025-25453111) funded by the Korean Government (MSIT), the Air Force Office of Scientific Research under Award No. FA23862514026, and Institute of Basic Science under project code IBS-R014-D1.

\vspace{1em}
\noindent{\bf Author Contributions}\\
W.L., H.K., and G.Y.C. conceived and designed the project. W.L. performed theoretical analyses and calculations under the supervision of H.K. and G.Y.C. All authors contributed to discussions and to writing and revising the manuscript.

\vspace{1em}  
\noindent {\bf Competing Interests}\\
The authors declare no competing interests.

\end{document}


\begin{center}
    {\color{gray}\large Supplementary Information for}
\end{center}

\title{Gluing Randomness via Entanglement: Tight Bound from Second R\'enyi Entropy}
\author{Wonjun Lee}
\email{wonjunphysics@gmail.com}
\affiliation{College of Natural Sciences, Korea Advanced Institute of Science and Technology, Daejeon, 34141, Korea}
\author{Hyukjoon Kwon}
\email{hjkwon@kias.re.kr}
\affiliation{School of Computational Sciences, Korea Institute for Advanced Study, Seoul, 02455, South Korea}
\author{Gil Young Cho}
\email{gilyoungcho@kaist.ac.kr}
\affiliation{Department of Physics, Korea Advanced Institute of Science and Technology, Daejeon, 34141, Korea}
\affiliation{Center for Artificial Low Dimensional Electronic Systems, Institute for Basic Science, Pohang, 37673, Republic of Korea}

\maketitle
\tableofcontents

\section*{Supplementary Notes}
\subsection*{1. Approximate state design from constant entanglement orbit (proof of Theorem 1)}
Let $\{A,B\}$ be a bipartition of a $n$-qubit system. For a $n$-qubit quantum state $\ket{\psi}$, we define the ensemble $\mathcal{E}_\mathrm{ent.}(\psi)$ of constantly entangled states as $\{U_A\otimes U_B\ket{\psi}|U_A\sim \mathcal{U}(\mathcal{H}_A),U_B\sim \mathcal{U}(\mathcal{H}_B)\}$. Here, $\mathcal{H}_A$ is the Hilbert space of $A$, and $\mathcal{U}(\mathcal{H}_A)$ is the Haar random unitary operators in $A$. We define $\Phi^{(t)}_{\mathcal{U}}(\rho)$ as $\mathbb{E}_{U\sim \mathcal{U}}[U^{\otimes t}\rho U^{\dag, \otimes t}]$. Additionally, for any positive integer $K$, we define $[K]^t_\mathrm{dist}$ as $\{\{k_i\}_{i=1}^t|\forall i, k_i\in[K], \nexists j\neq i \textrm{ such that } k_i=k_j\}$, \textit{i.e.}, the set of $t$ distinct integers from $[K]$. 
\begin{theorem}\label{thm:ent-bound}
    For any $n$-qubit state $\ket{\psi}$, $\mathcal{E}_\mathrm{ent.}(\psi)$ forms an approximate state $t$-design with error $O(t^2 e^{-\mathcal{N}_2(\psi)})$.
\end{theorem}
\begin{proof}
\noindent
Let the Schmidt decomposition of $\ket{\psi}$ be
\begin{equation}
    \ket{\psi} = \sum_{x\in[N_A]} \lambda_i \ket{x}_A\otimes\ket{x}_B
\end{equation}
with $N_A=2^{|A|}$. We define the projector $\Pi_\mathrm{dist}^{(t)}$ as
\begin{equation}
    \Pi_\mathrm{dist}^{(t)} = \sum_{\{x_i\}_{i=1}^t\in[N_A]^t_\mathrm{dist}}\ketbra{\{x_i\}}_A \otimes \ketbra{\{x_i\}}_B.
\end{equation}
By applying $\Pi_\mathrm{dist}^{(t)}$ to $t$-copies of $\ket{\psi}$, we get $|\psi^{(t)}_\mathrm{dist}\rangle$,
\begin{equation}
    |\psi^{(t)}_\mathrm{dist}\rangle = \Pi_\mathrm{dist}^{(t)}\ket{\psi}^{\otimes t}.
\end{equation}
The overlap between $|\psi^{(t)}_\mathrm{dist}\rangle$ and $\ket{\psi}^{\otimes t}$ is given by
\begin{equation}
    p_\mathrm{dist}=\langle\psi^{(t)}_\mathrm{dist}|(\ket{\psi})^{\otimes t} = \sum_{\{x_i\}_i^t\in[N_A]^t_\mathrm{dist}}\prod_{i=1}^t \lambda_{x_i}^2.
\end{equation}
Since $\lambda_{x_i}$ are positive, we can upper bound $1-p_\mathrm{dist}$ as
\begin{equation}
    1-p_\mathrm{dist} = \sum_{\{x_i\}_i^t\in[N_A]^t\backslash[N_A]^t_\mathrm{dist}}\prod_{i=1}^t \lambda_{x_i}^2 \leq \binom{t}{2}e^{-\mathcal{N}_2(\psi)}
\end{equation}
and $1-p_\mathrm{dist}^2$ as
\begin{equation}
    1-p^2_\mathrm{dist} \leq 2\binom{t}{2}e^{-\mathcal{N}_2(\psi)}.
\end{equation}
Thus, the trace distance between $|\psi^{(t)}_\mathrm{dist}\rangle$ and $\ket{\psi}^{\otimes t}$ is upper bounded by
\begin{equation}
    \begin{split}
        \mathrm{TD}\left(|\psi_\mathrm{dist}^{(t)}\rangle\langle\psi_\mathrm{dist}^{(t)}|,\ketbra{\psi}^{\otimes t}\right)
        &=\frac{1}{2}\sqrt{(\||\psi_\mathrm{dist}^{(t)}\rangle\|^2+\||\psi\rangle^{\otimes t}\|^2)^2-4|\langle\psi_\mathrm{dist}^{(t)}|\psi\rangle^{^\otimes t}|^2}\\
        &=\frac{1}{2}\sqrt{(1+p_\mathrm{dist}^2)^2-4 p_\mathrm{dist}^2}\\
        &=\frac{1}{2}(1-p_\mathrm{dist}^2)\\
        &\leq \binom{t}{2}e^{-\mathcal{N}_2(\psi)}\\
        &\leq \frac{t^2}{2}e^{-\mathcal{N}_2(\psi)}.
    \end{split}
\end{equation}
Next, let us consider an ensemble $\mathcal{E}_\mathrm{dist}$ of states made from $|\psi^{(t)}_\mathrm{dist}\rangle$ by applying unitary operators $U_A\otimes U_B$ sampled from Haar random unitary ensembles on $A$ and $B$. We note that this application does not change bi-partition entanglement of $|\psi^{(t)}_\mathrm{dist}\rangle$ in $A$ (or $B$). Then, the average state of this ensemble is given by
\begin{equation}
    \begin{split}
        \rho_\mathrm{con.ent.}^{(t)}
        &=\Phi^{(t)}_{\mathcal{U}(\mathcal{H}_A)}\otimes\Phi^{(t)}_{\mathcal{U}(\mathcal{H}_B)}\left(|\psi^{(t)}_\mathrm{dist}\rangle\langle\psi^{(t)}_\mathrm{dist}|\right)\\ 
        &= \sum_{\substack{\{x_i\}_{i=1}^t\in[N_A]^t_\mathrm{dist}\\\{y_i\}_{i=1}^t\in[N_A]^t_\mathrm{dist}}} \left(\prod_{i=1}^t\lambda_{x_i} \lambda_{y_i}\right)\Phi^{(t)}_{\mathcal{U}(\mathcal{H}_A)}(\ketbra{\{x_i\}}{\{y_i\}}_A)\otimes\Phi^{(t)}_{\mathcal{U}(\mathcal{H}_B)}(\ketbra{\{x_i\}}{\{y_i\}}_B).
    \end{split}
\end{equation}
We then use the approximate Haar twirl in Ref.~\cite{schuster2024} to get
\begin{equation}
    \begin{split}
        \rho_\mathrm{con.ent.}^{(t)}
        &=\rho_\mathrm{app.}^{(t)}+O(t^2/N_A)\\
        &=\frac{1}{2^{nt}}\sum_{\sigma,\sigma'\in S_t}\sum_{\substack{\{x_i\}_{i=1}^t\in[N_A]^t_\mathrm{dist}\\\{y_i\}_{i=1}^t\in[N_A]^t_\mathrm{dist}}} \left(\prod_{i=1}^t\lambda_{x_i}\lambda_{y_i}\right)\\
        &\qquad\qquad\times\tr(\ketbra{\{x_i\}}{\{y_i\}}_A P_{\sigma^{-1}}^{(A)})\\
        &\qquad\qquad\times\tr(\ketbra{\{x_i\}}{\{y_i\}}_B P_{\sigma'^{-1}}^{(B)})\\
        &\qquad\qquad\times P_\sigma^{(A)}\otimes P_{\sigma'}^{(B)}\\
        &\quad+O(t^2/N_A).
    \end{split}
\end{equation}
Here, $S_t$ is the symmetry group of order $t$, and $P_\sigma$ is the permutation operator that mixes $t$-copies according to $\sigma\in S_t$. We can simplify $\rho_\mathrm{app.}^{(t)}$ as
\begin{equation}
    \begin{split}
        \rho_\mathrm{app.}^{(t)}
        &=\frac{1}{2^{nt}}\sum_{\sigma,\sigma'\in S_t}\sum_{\{x_i\}_{i=1}^t\in[N_A]^t_\mathrm{dist}}\left(\prod_{i=1}^t\lambda_{x_i}^2\right)\delta_{\sigma(\{x_i\}),\sigma'(\{x_i\})}P_\sigma^{(A)}\otimes P_{\sigma'}^{(B)}\\
        &=\frac{1}{2^{nt}}\sum_{\sigma\in S_t}\sum_{\{x_i\}_{i=1}^t\in[N_A]^t_\mathrm{dist}}\left(\prod_{i=1}^t\lambda_{x_i}^2\right)P_\sigma^{(A)}\otimes P_{\sigma}^{(B)}\\
        &= \frac{t!}{2^{nt}}\sum_{\{x_i\}_{i=1}^t\in[N_A]^t_\mathrm{dist}}\left(\prod_{i=1}^t\lambda_{x_i}^2\right)\Pi^{(t)}_\mathrm{sym}\\
        &= \frac{t!}{2^{nt}}p_\mathrm{dist}\Pi^{(t)}_\mathrm{sym}\\
    \end{split}
\end{equation}
where $\Pi^{(t)}_\mathrm{sym}$ is the projector to the symmetry subspace of the $t$ copy space. This directly implies that
\begin{equation}
    \begin{split}
        \mathrm{TD}\left(\rho_\mathrm{app.}^{(t)},\rho^{(t)}_\mathrm{Haar}\right) 
        &= \frac{1}{2}\abs{\frac{t!}{2^{nt}}p_\mathrm{dist}\binom{N+t-1}{t}-1}\\
        &=\frac{1}{2}(1-p_\mathrm{dist}) + O(t^2/N_A)
    \end{split}
\end{equation}
with $N=2^n$. Using the convexity of the trace distance, we finally have
\begin{equation}
    \begin{split}
        \mathrm{TD}\left(\mathbb{E}_{\phi\sim\mathcal{E}}[\ketbra{\phi}^{\otimes t}],\rho^{(t)}_\mathrm{Haar}\right)
        &\leq\mathrm{TD}\left(\rho^{(t)}_\mathrm{con.ent.},\rho^{(t)}_\mathrm{Haar}\right) + \frac{t^2}{2}e^{-\mathcal{N}_2(\psi)}\\
        &\leq\mathrm{TD}\left(\rho^{(t)}_\mathrm{app.twl.},\rho^{(t)}_\mathrm{Haar}\right) + \frac{t^2}{2}e^{-\mathcal{N}_2(\psi)}+O(t^2/N_A)\\
        &\leq \frac{3}{4}t^2e^{-\mathcal{N}_2(\psi)}+O(t^2/N_A).
    \end{split}
\end{equation}
This gives the stated upper bound on the error $\epsilon$ as
\begin{equation}
    \epsilon 
    \leq \frac{3}{4}t^2 e^{-\mathcal{N}_2(\psi)} + O(t^2/N_A).
\end{equation}
We note that the maximum value of $\mathcal{N}_2(\psi)$ is $\log N_A$. This implies that $\epsilon$ is upper bounded by $O(t^2 e^{-\mathcal{N}_2(\psi)})$.
\end{proof}

\subsection*{2. Approximate state design from Markov chain (proof of Theorem 2)}
In this section, we proof \textbf{Theorem 2} in the main text. Before proving them, we first consider a simple tri-partite case.

Let the sequence $\mathfrak{C}=\{A,B,C\}$ be a tri-partition of a $n$-qubit system. We say a $n$-qubit state $\ket{\psi}$ with this sequence is a short quantum Markov chain if its conditional mutual information $I(A:C|B)$ is vanishing~\cite{Hayden2004}. A key property of this chain is that there exists a unitary operator $V_B$ supported in $B$ and a bipartition $\{B_A,B_C\}$ of $B$ such that $V_B\ket{\psi}=\ket{\phi}_{A\cup B_A}\ket{\varphi}_{B_C\cup C}$ for some states $\ket{\phi}_{A\cup B_A}$ and $\ket{\varphi}_{B_C\cup C}$. 

Let $\mathcal{E}_{\mathfrak{C},\mathrm{ent.}}(\psi)$ be the ensemble of states generated from $\ket{\psi}$ by applying Haar random unitary operators in each of the regions, \textit{i.e.}, $\mathcal{E}_{\mathfrak{C},\mathrm{ent.}}(\psi)=\{U_A\otimes U_B\otimes U_C\ket{\psi}|U_A\sim\mathcal{U}(\mathcal{H}_A),U_B\sim\mathcal{U}(\mathcal{H}_B),U_C\sim\mathcal{U}(\mathcal{H}_C)\}$. Then, the following holds.
\begin{theorem}
    For any $n$-qubit pure state $\ket{\psi}$, if $\ket{\psi}$ with a tri-partition $\{A,B,C\}$ is a short quantum Markov chain, then $\mathcal{E}_\mathrm{tri.ent.}(\psi)$ forms an approximate state $t$-design with error $O\left(t^2\left(\exp(-\mathcal{N}_2^{(A)}(\psi))+\exp(-\mathcal{N}_2^{(C)}(\psi))\right)\right)$.
\end{theorem}
\begin{proof}
Since $\ket{\psi}$ is a short quantum Markov chain, there exist a unitary operator $V_B$ supported in $B$, a bipartiion $\{B_A,B_C\}$ of $B$, $\ket{\phi}_{A\cup B_A}$, and $\ket{\varphi}_{B_C\cup C}$ such that
\begin{equation}
    |\tilde{\psi}\rangle = V_B\ket{\psi} = \ket{\phi}_{A\cup B_A}\ket{\varphi}_{B_C\cup C}.
\end{equation}
Let Schmidt decompositions of $\ket{\phi}_{A\cup B_A}$ and $\ket{\varphi}_{B_C\cup C}$ be
\begin{align}
    \ket{\phi}_{A\cup B_A} &= \sum_{x\in[2^{|A|}]} \lambda_x^{(A)} \ket{x}_A\ket{x}_{B_A}\\
    \ket{\varphi}_{B_C\cup C} &= \sum_{z\in[2^{|C|}]} \lambda_x^{(C)} \ket{z}_{B_C}\ket{z}_C
\end{align}
with $N_A=2^{|A|}$ and $N_C=2^{|C|}$. Let us define the following unnormalized state in the $t$-copy space:
\begin{equation}
    |\tilde{\psi}_\mathrm{dist}^{(t)}\rangle = \sum_{\{x_i\}_{i=1}^t\in[N_A]^t_\mathrm{dist}}\sum_{\{z_i\}_{i=1}^t\in[N_C]^t_\mathrm{dist}}\left(\prod_{i=1}^t\lambda^{(A)}_{x_i}\lambda^{(C)}_{z_i}\right)\ket{\{x_i\}_{i=1}^t}_A\ket{\{x_i\}_{i=1}^t}_{B_A}\ket{\{z_i\}_{i=1}^t}_{B_C}\ket{\{z_i\}_{i=1}^t}_C.
\end{equation}
The overlap between this state with $|\tilde{\psi}\rangle$ is given by
\begin{equation}
    \tilde{p}_\mathrm{dist}=\langle\tilde{\psi}_\mathrm{dist}^{(t)}|(|\tilde{\psi}\rangle^{\otimes t}) = \sum_{\{x_i\}_{i=1}^t\in[N_A]^t_\mathrm{dist}}\sum_{\{z_i\}_{i=1}^t\in[N_C]^t_\mathrm{dist}}\prod_{i=1}^t\left(\lambda^{(A)}_{x_i}\lambda^{(C)}_{z_i}\right)^2.
\end{equation}
This overlap can be bounded by
\begin{equation}
    1-\tilde{p}_\mathrm{dist}\leq \binom{t}{2}\left(e^{-\mathcal{N}_2^{(A)}(\psi)}+e^{-\mathcal{N}_2^{(C)}(\psi)}\right),
\end{equation}
which implies
\begin{equation}
    1-\tilde{p}_\mathrm{dist}^2 \leq 2\binom{t}{2}\left(e^{-\mathcal{N}_2^{(A)}(\psi)}+e^{-\mathcal{N}_2^{(C)}(\psi)}\right).
\end{equation}
Using this, we can bound the trace distance between $|\tilde{\psi}_\mathrm{dist}^{(t)}\rangle$ and $|\tilde{\psi}\rangle^{\otimes t}$ as
\begin{equation}
    \mathrm{TD}\left(|\tilde{\psi}^{(t)}_\mathrm{dist}\rangle\langle\tilde{\psi}^{(t)}_\mathrm{dist}|, |\tilde{\psi}\rangle\langle\tilde{\psi}|^{\otimes t}\right) \leq \frac{t^2}{2}\left(e^{-\mathcal{N}_2^{(A)}(\psi)}+e^{-\mathcal{N}_2^{(C)}(\psi)}\right).
\end{equation}
Then, the $t$-th moments of $\mathcal{E}_\mathrm{tri.ent.}$ can be approximated using the approximate Haar twirl in Ref.~\cite{schuster2024} as follows:
\begin{equation}
    \begin{split}
        \rho^{(t)}
        &=\mathbb{E}_{\theta\sim\mathcal{E}_\mathrm{tri.ent.}}[\ketbra{\theta}^{\otimes t}]\\
        &=\Phi^{(t)}_{\mathcal{U}(\mathcal{H}_A)}\otimes \Phi^{(t)}_{\mathcal{U}(\mathcal{H}_B)}\otimes \Phi^{(t)}_{\mathcal{U}(\mathcal{H}_C)} (|\psi\rangle\langle\psi|^{\otimes t})\\
        &= \Phi^{(t)}_{\mathcal{U}(\mathcal{H}_A)}\otimes \Phi^{(t)}_{\mathcal{U}(\mathcal{H}_B)}\otimes \Phi^{(t)}_{\mathcal{U}(\mathcal{H}_C)} (|\tilde{\psi}\rangle\langle\tilde{\psi}|^{\otimes t})\\
        &= \Phi^{(t)}_{\mathcal{U}(\mathcal{H}_A)}\otimes \Phi^{(t)}_{\mathcal{U}(\mathcal{H}_B)}\otimes \Phi^{(t)}_{\mathcal{U}(\mathcal{H}_C)} (|\tilde{\psi}^{(t)}_\mathrm{dist}\rangle\langle\tilde{\psi}^{(t)}_\mathrm{dist}|) + O\left(t^2(e^{-\mathcal{N}_2^{(A)}(\psi)}+e^{-\mathcal{N}_2^{(C)}(\psi)})\right)\\
        &= \sum_{\{x_i\}_{i=1}^t,\{y_i\}_{i=1}^t\in[N_A]^t_\mathrm{dist}}\sum_{\{z_i\}_{i=1}^t,\{w_i\}_{i=1}^t\in[N_C]^t_\mathrm{dist}}\left(\prod_{i=1}^t\lambda^{(A)}_{x_i}\lambda^{(A)}_{y_i}\lambda^{(C)}_{z_i}\lambda^{(C)}_{w_i}\right)\Phi^{(t)}_{\mathcal{U}(\mathcal{H}_A)}\left(\ketbra{\{x_i\}_{i=1}^t}{\{y_i\}_{i=1}^t}_A\right)\\
        &\qquad\qquad\otimes\Phi^{(t)}_{\mathcal{U}(\mathcal{H}_B)}\left(\ketbra{\{x_i\}_{i=1}^t}{\{y_i\}_{i=1}^t}_{B_A}\otimes\ketbra{\{z_i\}_{i=1}^t}{\{w_i\}_{i=1}^t}_{B_C}\right)\otimes\Phi^{(t)}_{\mathcal{U}(\mathcal{H}_C)}\left(\ketbra{\{z_i\}_{i=1}^t}{\{w_i\}_{i=1}^t}_C\right)\\
        &\quad+O\left(t^2(e^{-\mathcal{N}_2^{(A)}(\psi)}+e^{-\mathcal{N}_2^{(C)}(\psi)})\right)\\
        &=\frac{1}{2^{nt}}\sum_{\sigma,\tau,\pi\in S_t}\sum_{\{x_i\}_{i=1}^t,\{y_i\}_{i=1}^t\in[N_A]^t_\mathrm{dist}}\sum_{\{z_i\}_{i=1}^t,\{w_i\}_{i=1}^t\in[N_C]^t_\mathrm{dist}}\left(\prod_{i=1}^t\lambda^{(A)}_{x_i}\lambda^{(A)}_{y_i}\lambda^{(C)}_{z_i}\lambda^{(C)}_{w_i}\right)\tr(\ketbra{\{x_i\}_{i=1}^t}{\{y_i\}_{i=1}^t}_A P_{\sigma^{-1}}^{(A)})\\
        &\qquad\qquad\times\tr\left(\ketbra{\{x_i\}_{i=1}^t}{\{y_i\}_{i=1}^t}_{B_A}\otimes\ketbra{\{z_i\}_{i=1}^t}{\{w_i\}_{i=1}^t}_{B_C}P_{\tau^{-1}}^{(B)}\right)\tr\left(\ketbra{\{z_i\}_{i=1}^t}{\{w_i\}_{i=1}^t}_C P_{\pi^{-1}}^{(C)}\right)\\
        &\qquad\qquad\times P_\sigma^{(A)}\otimes P_\tau^{(B)}\otimes P_\pi^{(C)}+O\left(t^2(e^{-\mathcal{N}_2^{(A)}(\psi)}+e^{-\mathcal{N}_2^{(C)}(\psi)})\right)\\
        &=\frac{t!}{2^{nt}}\tilde{p}_\mathrm{dist}\Pi^{(t)}_\mathrm{sym}+ O\left(t^2(e^{-\mathcal{N}_2^{(A)}(\psi)}+e^{-\mathcal{N}_2^{(C)}(\psi)})\right).
    \end{split}
\end{equation}
Here, we use the unitary invariance of $\Phi^{(t)}_{\mathcal{U}(\mathcal{H}_B)}$ at the second equality, and $P_\sigma^{(A)}$ represents the copy-wise $\sigma$-permutation operator acting on $A$. Therefore, we conclude that
\begin{equation}
    \mathrm{TD}\left(\rho^{(t)},\rho^{(t)}_\mathrm{Haar}\right) \leq O\left(t^2(e^{-\mathcal{N}_2^{(A)}(\psi)}+e^{-\mathcal{N}_2^{(C)}(\psi)})\right).
\end{equation}
\end{proof}

One may observe that entanglement between $A$ and $B$ effectively connects the permutation operators generated by $\Phi^{(t)}_{\mathcal{U}(\mathcal{H}_A)}$ to subsystem $B$. Consequently, the permutation operators generated by $\Phi^{(t)}_{\mathcal{U}(\mathcal{H}_B)}$ coincide with those of $A$ due to the orthogonality of the permutations. Extending this logic to the permutations generated by $\Phi^{(t)}_{\mathcal{U}(\mathcal{H}_C)}$ results in the formation of the global permutation operators acting on $ABC$. The same logic can be applied to a generic pure quantum Markov chain for some integer $l$ as follows.

\begin{theorem}
    For any pure quantum Markov chain $\ket{\psi}_{\mathfrak{C}}$ with a sequence $\mathfrak{C}$ of disjointed regions $\{A_i\}_{i=1}^l$ and $l\geq 3$, $\mathcal{E}_{\mathfrak{C},\mathrm{ent.}}(\psi)$ form an approximate state $t$-design with error $O\left(t^2\sum_{i=1}^{l-1} e^{-\mathcal{N}_2^{(\mathcal{A}_i)}(\psi)}\right)$ and $\mathcal{A}_i=\bigcup_{j=1}^{i}A_j$.
\end{theorem}
\begin{proof}
Since $\ket{\psi}_\mathfrak{C}$ is a quantum Markov chain, there exist $\{V_{A_i}\}_{i=2}^{l-1}$ and bipartitions $\{(A_{i,L},A_{i,R})\}_{i=2}^{l-1}$ of $\{A_i\}_{i=2}^{l-1}$ such that 
\begin{equation}
    |\tilde{\psi}\rangle_\mathfrak{C} = \bigotimes_{i=2}^{l-1} V_{A_i} \ket{\psi}_\mathfrak{C} = \ket{\phi_1}_{A_1\cup A_{2,L}} \otimes \left( \bigotimes_{i=2}^{l-2} \ket{\phi_i}_{A_{i,R}\cup A_{i+1,L}} \right) \otimes \ket{\phi_{l-1}}_{A_{l-1,R}\cup A_l}
\end{equation}
for some states $\{\ket{\phi_1}_{A_1\cup A_{2,L}},\ket{\phi_{l-1}}_{A_{l-1,R}\cup A_l}\}\cup\{\ket{\phi_i}_{A_{i,R}\cup A_{i+1,L}}\}_{i=2}^{l-2}$. Due to the unitary invariance of the Haar measure, $\mathcal{E}_{\mathfrak{C},\mathrm{ent.}}(\psi)$ and $\mathcal{E}_{\mathfrak{C},\mathrm{ent.}}(\tilde{\psi})$ are identical. Therefore, we will consider $\mathcal{E}_{\mathfrak{C},\mathrm{ent.}}(\tilde{\psi})$ instead of $\mathcal{E}_{\mathfrak{C},\mathrm{ent.}}(\psi)$. Let Schmidt decompositions of $\ket{\phi_i}_{A_{i,R}\cup A_{i+1,L}}$ be 
\begin{equation}
    \ket{\phi_i}_{A_{i,R}\cup A_{i+1,L}} = \sum_{x\in\left[2^{|A_{i,R}|}\right]} \lambda_x^{(i)} \ket{x}_{A_{i,R}}\otimes\ket{x}_{A_{i+1,L}}.
\end{equation}
We note that the Schmidt values $\{\lambda_x^{(i)}\}$ are equivalent to those of the bipartition $(\mathcal{A}_i,\mathcal{A}_i^c)$. We define $p^{(i)}_\mathrm{dist}$ as
\begin{equation}
    p^{(i)}_\mathrm{dist}=\sum_{\{x_i\}^t_{k=1}\in\left[2^{|A_{i+1,L}|}\right]^t_\mathrm{dist}}\prod_{k=1}^t (\lambda_{x_k}^{(i)})^2.
\end{equation}
Then, it is bounded by
\begin{equation}
    1-p^{(i)}_\mathrm{dist} \leq t^2 e^{-\mathcal{N}_2^{(\mathcal{A}_i)}(\psi)}.
\end{equation}

Now, for $j<l$, we claim that the first $j$-th twirl operators acting on the first $j$-th states give
\begin{equation}\label{eq:gluing-ansatz}
    \begin{split}
        &\bigotimes_{i=1}^j \Phi^{(t)}_{\mathcal{U}(\mathcal{H}_{A_i})}\left( \ketbra{\phi_1}^{\otimes t}_{A_1\cup A_{2,L}} \otimes \bigotimes_{i=2}^j \ketbra{\phi_i}^{\otimes t}_{A_{i,R}\cup A_{i+1,L}} \right)\\
        &= \frac{1}{2^{|\mathcal{A}_j|t}}\prod_{i=1}^{j-1} p_\mathrm{dist}^{(i)}\sum_{\sigma\in S_t}\sum_{\{x_i\}_{k=1}^t\in\left[2^{|A_{j+1,L}|}\right]^t_\mathrm{dist}} \prod_{k=1}^t (\lambda_{x_k}^{(j)})^2 P_\sigma^{(\mathcal{A}_j)}\otimes \ketbra{\{x_k\}^{t}_{k=1}}{\sigma^{-1}(\{x_k\}^{t}_{k=1})}_{A_{j+1,L}}\\
        &\qquad + O\left(t^2\sum_{i=1}^j e^{-\mathcal{N}_2^{(\mathcal{A}_i)}(\psi)}\right).
    \end{split}
\end{equation}
Here, the term given by the big $O$ notation represents the error in the trace distance. We show this using the mathematical induction. First, when $j=1$, we have
\begin{equation}
    \begin{split}
        &\Phi^{(t)}_{\mathcal{U}(\mathcal{H}_{A_1})}\left(\ketbra{\phi_1}_{A_1\cup A_{2,L}}^{\otimes t}\right) \\
        &= \sum_{\{x_k\}_{k=1}^t,\{y_k\}_{k=1}^t\in[2^{|A_{2,L}|}]^t_\mathrm{dist}}\left(\prod_{k=1}^t\lambda_{x_k}^{(1)}\lambda_{y_k}^{(1)}\right)\Phi^{(t)}_{\mathcal{U}(\mathcal{H}_{A_1})}\left(\ketbra{\{x_k\}_{k=1}^t}{\{y_k\}_{k=1}^t}_{A_1}\right)\otimes \ketbra{\{x_i\}_{k=1}^t}{\{y_k\}_{k=1}^t}_{A_{2,L}}\\
        &\qquad + O\left(t^2 e^{-\mathcal{N}_2^{(\mathcal{A}_1)}(\psi)} \right) \\
        &= \frac{1}{2^{|\mathcal{A}_1|t}}\sum_{\{x_k\}_{k=1}^t,\{y_k\}_{k=1}^t\in[2^{|A_{2,L}|}]^t_\mathrm{dist}}\left(\prod_{k=1}^t\lambda_{x_k}^{(1)}\lambda_{y_k}^{(1)}\right)\sum_{\sigma\in S_t}\tr\left(\ketbra{\{x_k\}_{k=1}^t}{\{y_k\}_{k=1}^t}_{A_1} P_{\sigma^{-1}}^{(A_1)}\right)\\
        &\qquad\qquad\qquad\qquad\qquad\qquad\qquad\qquad\qquad\qquad\qquad\qquad \times P_{\sigma}^{(A_1)}\otimes \ketbra{\{x_k\}_{k=1}^t}{\{y_k\}_{k=1}^t}_{A_{2,L}}\\
        &\qquad + O\left(t^2 e^{-\mathcal{N}_2^{(\mathcal{A}_1)}(\psi)} \right)\\
        &= \frac{1}{2^{|\mathcal{A}_1|t}}\sum_{\sigma\in S_t}\sum_{\{x_k\}_{k=1}^t\in[2^{|A_{2,L}|}]^t_\mathrm{dist}}\prod_{k=1}^t(\lambda_{x_k}^{(1)})^2 P_\sigma^{(A_1)}\otimes\ketbra{\{x_k\}_{k=1}^t}{\sigma^{-1}(\{x_k\}_{k=1}^t)}_{A_{2,L}}+O\left(t^2 e^{-\mathcal{N}_2^{(\mathcal{A}_1)}(\psi)} \right).
    \end{split}
\end{equation}
Now, let us assume that Eq.~\eqref{eq:gluing-ansatz} holds for $j=m-1$. Then, we have
\begin{equation}
    \begin{split}
        &\bigotimes_{i=1}^m \Phi^{(t)}_{\mathcal{U}(\mathcal{H}_{A_i})}\left( \ketbra{\phi_1}^{\otimes t}_{A_1\cup A_{2,L}} \otimes \bigotimes_{i=2}^m \ketbra{\phi_i}^{\otimes t}_{A_{i,R}\cup A_{i+1,L}} \right)\\
        &= \frac{1}{2^{|\mathcal{A}_{m-1}|t}}\prod_{i=1}^{m-1}p_\mathrm{dist}^{(i)}\sum_{\sigma\in S_t}\sum_{\{x_k\}_{k=1}^t\in[2^{|A_{m,L}|}]^t_\mathrm{dist}}\prod_{i=1}^t(\lambda_{x_k}^{(m-1)})^2 P_\sigma^{(\mathcal{A}_{m-1})}\\
        &\qquad\qquad\qquad\qquad\qquad\qquad\qquad\otimes\Phi_{\mathcal{U}(\mathcal{H}_{A_m})}^{(t)}\left( \ketbra{\{x_k\}_{k=1}^t}{\sigma^{-1}(\{x_k\}_{k=1}^t)}_{A_{m,L}}\otimes \ketbra{\phi_m}^{\otimes t}_{A_{m,R}\cup A_{m+1,L}}\right) \\
        &\quad + O\left(t^2\sum_{i=1}^{m-1} e^{-\mathcal{N}_2^{(\mathcal{A}_i)}(\psi)}\right)\\
        &= \frac{1}{2^{|\mathcal{A}_{m-1}|t}}\prod_{i=1}^{m-1}p_\mathrm{dist}^{(i)}\sum_{\sigma\in S_t}\sum_{\{x_k\}_{k=1}^t\in[2^{|A_{m,L}|}]^t_\mathrm{dist}}\sum_{\{z_k\}_{k=1}^t,\{w_k\}_{k=1}^t\in[2^{|A_{m+1,L}|}]^t_\mathrm{dist}}\left(\prod_{i=1}^t(\lambda_{x_k}^{(m-1)})^2\lambda_{z_k}^{(m)}\lambda_{w_k}^{(m)}\right)\\
        &\qquad\qquad\qquad\qquad\qquad\qquad\qquad \times P_\sigma^{(\mathcal{A}_{m-1})} \otimes \ketbra{\{z_k\}_{k=1}^t}{\{w_k\}_{k=1}^t}_{A_{m+1,L}}\\
        &\qquad\qquad\qquad\qquad\qquad\qquad\qquad \otimes \Phi_{\mathcal{U}(\mathcal{H}_{A_m})}^{(t)}\left( \ketbra{\{x_k\}_{k=1}^t}{\sigma^{-1}(\{x_k\}_{k=1}^t)}_{A_{m,L}} \otimes \ketbra{\{z_k\}_{k=1}^t}{\{w_k\}_{k=1}^t}_{A_{m,R}}\right)\\
        &\quad + O\left(t^2\sum_{i=1}^m e^{-\mathcal{N}_2^{(\mathcal{A}_i)}(\psi)}\right)\\
        &= \frac{1}{2^{|\mathcal{A}_m|t}}\prod_{i=1}^{m-1}p_\mathrm{dist}^{(i)}\sum_{\sigma,\tau\in S_t}\sum_{\{x_k\}_{k=1}^t\in[2^{|A_{m,L}|}]^t_\mathrm{dist}}\sum_{\{z_k\}_{k=1}^t,\{w_k\}_{k=1}^t\in[2^{|A_{m+1,L}|}]^t_\mathrm{dist}}\left(\prod_{k=1}^t(\lambda_{x_k}^{(m-1)})^2\lambda_{z_k}^{(m)}\lambda_{w_k}^{(m)}\right)\\
        &\qquad\qquad\qquad\qquad\qquad\qquad\qquad\qquad\times\tr(\ketbra{\{x_k\}_{k=1}^t}{\sigma^{-1}(\{x_k\}_{k=1}^t)}_{A_{m,R}}\otimes\ketbra{\{z_k\}_{k=1}^t}{\{w_k\}_{k=1}^t}_{A_{m+1,L}} P_{\tau^{-1}}^{(A_m)})\\
        &\qquad\qquad\qquad\qquad\qquad\qquad\qquad\qquad\times P_\sigma^{(\mathcal{A}_i)}\otimes P_\tau^{(A_m)}\otimes\ketbra{\{z_k\}_{k=1}^t}{\{w_k\}_{k=1}^t}_{A_{m+1,L}}\\
        &\quad + O\left(t^2\sum_{i=1}^m e^{-\mathcal{N}_2^{(\mathcal{A}_i)}(\psi)}\right)\\
        &=\frac{1}{2^{|\mathcal{A}_m|t}}\prod_{i=1}^m p_\mathrm{dist}^{(i)}\sum_{\sigma\in S_t}\sum_{\{z_k\}_{k=1}^t\in[2^{|A_{m+1,L}|}]^t_\mathrm{dist}}\prod_{k=1}^t(\lambda_{z_k}^{(m)})^2 P_\sigma^{(\mathcal{A}_m)}\otimes\ketbra{\{z_k\}_{k=1}^t}{\sigma^{-1}(\{z_k\}_{k=1}^t)}_{A_{m+1,L}}\\
        &\quad + O\left(t^2\sum_{i=1}^m e^{-\mathcal{N}_2^{(\mathcal{A}_i)}(\psi)}\right).
    \end{split}
\end{equation}
This is equivalent to Eq.~\eqref{eq:gluing-ansatz} with $j=m$. Then, due to the mathematical induction, Eq.~\eqref{eq:gluing-ansatz} holds for all $1\leq j < l$. Using this, we can compute the $t$-th moments of $\mathcal{E}_{\mathfrak{C},\mathrm{ent.}}(\tilde{\psi})$ as follows:
\begin{equation}
    \begin{split}
        \rho^{(t)}
        &=\mathbb{E}_{\theta\sim\mathcal{E}_{\mathfrak{C},\mathrm{ent.}}(\tilde{\psi})}[\ketbra{\theta}^{\otimes t}]\\
        &= \bigotimes_{i=1}^l \Phi^{(t)}_{\mathcal{U}(\mathcal{H}_{A_i})}\left( \ketbra{\phi_1}^{\otimes t}_{A_1\cup A_{2,L}} \otimes\left(\bigotimes_{i=2}^{l-2} \ketbra{\phi_i}^{\otimes t}_{A_{i,R}\cup A_{i+1,L}}\right) \otimes \ketbra{\phi_{l-1}}^{\otimes t}_{A_{l-1,R}\cup A_l} \right) \\
        &= \frac{1}{2^{|\mathcal{A}_{l-1}|t}}\prod_{i=1}^{l-2}p_\mathrm{dist}^{(i)}\sum_{\sigma\in S_t}\sum_{\{x_k\}_{k=1}^t\in[2^{|A_l|}]^t_\mathrm{dist}}\prod_{i=1}^t(\lambda_{x_k}^{(l-1)})^2 P_\sigma^{(\mathcal{A}_{l-1})}\\
        &\qquad\qquad\qquad\qquad\qquad\qquad\qquad\otimes\Phi_{\mathcal{U}(\mathcal{H}_{A_l})}^{(t)}\left( \ketbra{\{x_k\}_{k=1}^t}{\sigma^{-1}(\{x_k\}_{k=1}^t)}_{A_l}\right) \\
        &\quad + O\left(t^2\sum_{i=1}^l e^{-\mathcal{N}_2^{(\mathcal{A}_i)}(\psi)}\right)\\
        &= \frac{1}{2^{nt}}\prod_{i=1}^{l-2}p_\mathrm{dist}^{(i)}\sum_{\sigma,\tau\in S_t}\sum_{\{x_k\}_{k=1}^t\in[2^{|A_l|}]^t_\mathrm{dist}}\prod_{i=1}^t(\lambda_{x_k}^{(l-1)})^2 P_\sigma^{(\mathcal{A}_{l-1})}\otimes P_\tau^{(A_l)}\\
        &\qquad\qquad\qquad\qquad\qquad\qquad\qquad\times\tr\left( \ketbra{\{x_k\}_{k=1}^t}{\sigma^{-1}(\{x_k\}_{k=1}^t)}_{A_l} P_{\tau^{-1}}^{(A_l)}\right) \\
        &= \frac{1}{2^{nt}}\prod_{i=1}^{l-1}p_\mathrm{dist}^{(i)}\sum_{\sigma\in S_t} P_\sigma + O\left(t^2\sum_{i=1}^l e^{-\mathcal{N}_2^{(\mathcal{A}_i)}(\psi)}\right).
    \end{split}
\end{equation}
Therefore, the $t$-th moments of $\mathcal{E}_{\mathfrak{C},\mathrm{ent.}}(\tilde{\psi})$ approximates $\rho^{(t)}_\mathrm{Haar}$ as
\begin{equation}
    \begin{split}
        \mathrm{TD}\left(\rho^{(t)}, \rho^{(t)}_\mathrm{Haar}\right) 
        &= \frac{1}{2}\left(1-\prod_{i=1}^{l-1}p^{(i)}_\mathrm{dist}\right) + O\left(t^2\sum_{i=1}^l e^{-\mathcal{N}_2^{(\mathcal{A}_i)}(\psi)}\right)\\
        &\leq O\left(t^2\sum_{i=1}^l e^{-\mathcal{N}_2^{(\mathcal{A}_i)}(\psi)}\right).
    \end{split}
\end{equation}
\end{proof}

\subsection*{3. Approximate state design from GHZ state (proof of Theorem 3)}
Here, we proof \textbf{Theorem 3} in the main text. 
\begin{theorem}
    For any set $\mathfrak{C}$ of disjointed regions $\{A_i\}_{i=1}^l$ and a positive integer $d\leq \min_{1\leq i\leq l}(2^{|A_i|})$, $\mathcal{E}_{\mathfrak{C},\mathrm{ent.}}(\psi)$ generated from the $d$-level GHZ state $\ket{\psi}=\frac{1}{\sqrt{d}}\sum_{x=1}^d\bigotimes_{i=1}^l\ket{x}_{A_{i}}$ forms an approximate state $t$-design with error $O\left(t^2/d\right)$.
\end{theorem}
\begin{proof}
    We first define the state obtained by $\ket{\psi}^{\otimes t}$ projecting to copy-wise distinct configurations.
    \begin{equation}
        |\psi_\mathrm{dist}^{(t)}\rangle = \frac{1}{\sqrt{d^t}}\sum_{\{x_k\}_{k=1}^t\in[d]^t_\mathrm{dist}}\bigotimes_{i=1}^l \ket{\{x_k\}_{k=1}^t}_{A_i}
    \end{equation}
    This state has overlap with $\ket{\psi}^{\otimes t}$ as
    \begin{equation}
        p_\mathrm{dist} = \langle\psi_\mathrm{dist}^{(t)}|(\ket{\psi})^{\otimes t} = \frac{d!}{d^t (d-t)!}.
    \end{equation}
    This overlap is bounded by
    \begin{equation}
        1-p_\mathrm{dist} \leq \frac{t^2}{d}.
    \end{equation}
    Therefore, their trace distance is upper bounded by
    \begin{equation}
        \mathrm{TD}\left(\ketbra{\psi}^{\otimes t}, |\psi_\mathrm{dist}^{(t)}\rangle\langle \psi_\mathrm{dist}^{(t)}|\right) \leq O\left(t^2/d\right).
    \end{equation}
    Now, let us compute the ensemble average of $\ket{\psi}^{\otimes t}$. As shown below, this can be approximated as the symmetry projector.
    \begin{equation}
        \begin{split}
            \bigotimes_{i=1}^l \Phi^{(t)}_{\mathcal{U}(\mathcal{H}_{A_i})}\left(|\psi_\mathrm{dist}^{(t)}\rangle\langle\psi_\mathrm{dist}^{(t)}|^{\otimes t}\right) 
            &= \frac{1}{d^t}\sum_{\{x_k\}_{k=1}^t,\{y_k\}_{k=1}^t\in[d]^t_\mathrm{dist}}\bigotimes_{i=1}^l \Phi^{(t)}_{\mathcal{U}(\mathcal{H}_{A_i})}\left( \ketbra{\{x_k\}_{k=1}^t}{\{y_k\}_{k=1}^t}_{A_i}\right)\\
            &= \frac{1}{d^t 2^{nt}}\sum_{\sigma_1,\cdots,\sigma_l\in S_t}\sum_{\{x_k\}_{k=1}^t,\{y_k\}_{k=1}^t\in[d]^t_\mathrm{dist}}\bigotimes_{i=1}^l P^{(A_i)}_{\sigma_i} \\
            &\qquad\qquad\qquad\times\prod_{i=1}^l\tr(\ketbra{\{x_k\}_{k=1}^t}{\{y_k\}_{k=1}^t}_{A_i} P_{\sigma_i^{-1}}^{(A_i)})\\
            &\quad+O(l t^2/2^n)\\
            &= \frac{d!}{d^t (d-t)!} \frac{1}{2^{nt}} \sum_{\sigma\in S_t}P_\sigma +O\left( \frac{lt^2}{2^n}\right)
        \end{split}
    \end{equation}
    This implies that the $t$-th moments of $\mathcal{E}_{\mathfrak{C},\mathrm{ent.}}(\psi)$ has the trace distance from $\rho^{(t)}_\mathrm{Haar}$ as
    \begin{equation}
        \mathrm{TD}\left(\rho^{(t)},\rho^{(t)}_\mathrm{Haar}\right)\leq O(t^2/d).
    \end{equation}
\end{proof}

\subsection*{4. Approximate state design from constant coherence orbit (proof of Lemma 4)}
For a $n$-qubit state $\ket{\psi}$, we define the ensemble $\mathcal{E}_\mathrm{coh.}(\psi)$ of constantly coherent states as $\{PF\ket{\psi}| P\in\mathcal{P},F\in\mathcal{F}\}$. Here, $\mathcal{P}$ and $\mathcal{F}$ are set of all permutation and phase operators.
\begin{theorem}\label{thm:coh-bound}
    For any $n$-qubit state $\ket{\psi}$, $\mathcal{E}_\mathrm{coh.}(\psi)$ forms an approximate state $t$-design with error $O(t^2 e^{-\mathcal{C}_2(\psi)})$.
\end{theorem}
\begin{proof}
\noindent
Let us write $\ket{\psi}$ in the computational basis,
\begin{equation}
    \ket{\psi} = \sum_{x\in[N]} c_x \ket{x}.
\end{equation}
The ensemble average of $t$ copies of $U_f\ket{\psi}$ over $\mathcal{F}$ is given by
\begin{equation}
    \begin{split}
        \rho_\mathrm{F}^{(t)}
        &= \Phi_\mathcal{F}\left(\ketbra{\psi}^{\otimes t}\right)\\
        &= \sum_{\{x_i,y_i\}_{i=1}^t\in [N]^{2t}} \left(\prod_{i=1}^t c_{x_i} c^*_{y_i}\right)\mathbb{E}_{f}\left[e^{i\sum_{j=1}^t[f(x_i)+f(y_i)]}\right]\ketbra{x_1,\cdots,x_t}{y_1,\cdots,y_t}\\
        &= \sum_{\substack{\{x_i,y_i\}_{i=1}^t\in [N]^{2t}\\ \{x_i\}=\{y_i\}}} \left(\prod_{i=1}^t c_{x_i} c^*_{y_i}\right)\ketbra{x_1,\cdots,x_t}{y_1,\cdots,y_t}.
    \end{split}
\end{equation}
Here, we use the type vector $\mathrm{type}(v)$ introduced in Refs.~\cite{harrow2013church, ananth2022function}. Let us compute the probability $p_\mathrm{dist}$ of $\mathrm{type}(v)$ in $\mathcal{U}_{n,t}$. This probability is given by
\begin{equation}
    p_\mathrm{dist} = \sum_{\{x_i\}_{i=1}^t\in[N]^t_\mathrm{dist}}\prod_{i=1}^t\abs{c_{x_i}}^2.
\end{equation}
We can upper bound $1-p_\mathrm{dist}$ as
\begin{equation}
    1-p_\mathrm{dist} = \sum_{\{x_i\}_i^t\in[N]^t\backslash[N]^t_\mathrm{dist}}\prod_{i=1}^t \abs{c_{x_i}}^2 \leq \binom{t}{2}e^{-\mathcal{C}_2(\psi)}.
\end{equation}
Let us define the projector $\Pi^{(t)}_{\mathrm{unique},n}$ to the space spanned by unique type states:
\begin{equation}
    \Pi^{(t)}_{\mathrm{unique},n} = \sum_{\{x_i\}_i^t\in[N]^t\backslash[N]^t_\mathrm{dist}}\ketbra{\{x_i\}}.
\end{equation}
Then, we have
\begin{equation}
    \mathrm{TD}\left(\rho_\mathrm{F}^{(t)}, \Pi^{(t)}_{\mathrm{unique},n}\rho_\mathrm{F}^{(t)}\right) = 1-p_\mathrm{dist} \leq \binom{t}{2}e^{-\mathcal{C}_2(\psi)}.
\end{equation}
Next, we compute the ensemble average of $\Pi^{(t)}_{\mathrm{unique},n}\rho_\mathrm{F}^{(t)}$ conjugated by $U_p^{\otimes t}$ over $\mathcal{P}$:
\begin{equation}
    \begin{split}
        \rho_\mathrm{PF}^{(t)}
        &= \Phi_\mathcal{P}\left(\Pi^{(t)}_{\mathrm{unique},n}\rho_\mathrm{F}^{(t)}\right) \\
        &= p_\mathrm{dist} \binom{N}{t}^{-1}\sum_{\mathrm{type}(v)\in\mathcal{U}_{S,t}} \ketbra{\mathrm{type}(v)}{\mathrm{type}(v)}.
    \end{split}
\end{equation}
We define the equal sum of the unique type states as
\begin{equation}
    \rho_{\mathrm{unique},n}^{(t)} = \binom{N}{t}^{-1} \sum_{\mathrm{type}(v)\in\mathcal{U}_{S,t}} \ketbra{\mathrm{type}(v)}{\mathrm{type}(v)}.
\end{equation}
Then, the trace distance between $\rho^{(t)}_\mathrm{PF}$ and $\rho_{\mathrm{unique},n}^{(t)}$ is upper bounded by
\begin{equation}
    \begin{split}
        \mathrm{TD}\left(\rho^{(t)}_\mathrm{PF},\rho_{\mathrm{unique},n}^{(t)}\right) = 1-p_\mathrm{dist} \leq \binom{t}{2}e^{-\mathcal{C}_2(\psi)}.
    \end{split}
\end{equation}
Therefore, we finally have
\begin{equation}
    \begin{split}
        \mathrm{TD}\left(\mathbb{E}_{\phi\sim\mathcal{E}_{\mathrm{coh.}}(\psi)}[\ketbra{\phi}^{\otimes t}], \rho^{(t)}_{\mathrm{Haar},n}\right)
        &\leq \mathrm{TD}\left(\rho^{(t)}_\mathrm{PF}, \rho^{(t)}_{\mathrm{Haar},n}\right) + \binom{t}{2}e^{-\mathcal{C}_2(\psi)}\\
        &\leq \mathrm{TD}\left(\rho^{(t)}_{\mathrm{unique},n},\rho^{(t)}_{\mathrm{Haar},n}\right) + t^2e^{-\mathcal{C}_2(\psi)}\\
        &\leq t^2e^{-\mathcal{C}_2(\psi)} + O(t^2/N),
    \end{split}
\end{equation}
which gives
\begin{equation}
    \epsilon \leq t^2 e^{-\mathcal{C}_2(\psi)} + O(t^2/N).
\end{equation}
Since the maximum value of $\mathcal{C}_2(\psi)$ is $\log N$, $\epsilon$ is upper bounded by $O(t^2 e^{-\mathcal{C}_2(\psi)})$.
\end{proof}

\subsection*{5. Maximal achievable Haar randomness (proof of Theorem 6)}
Here, we study the maximal achievable Haar randomness, measured by the trace distance from the Haar random states, without using additional resourceful gates. The following theorem states that among ensembles on constant resourceful states, the ensemble with the uniform distribution is closest to the Haar random states. 
\begin{theorem}\label{thm:uniform-maximal}
    Let $G$ be a group of unitary operators. For any state $\ket{\psi}$, an ensemble $\mathcal{E}'(\psi)$ constructed by applying unitary operators in $G$ to $\ket{\psi}$ according to any distribution cannot be closer to the Haar random states in the trace distance than the ensemble $\mathcal{E}_\mathrm{unif.}(\psi)$ with a uniform distribution.
\end{theorem}
\begin{proof}
    Let $\mathcal{U}$ be the ensemble of unitary operators in $G$ with the uniform distribution. Since $\mathcal{U}$ follows a uniform distribution, we have
    \begin{equation}
        \mathbb{E}_{\phi\sim\mathcal{E}(\psi)}[\ketbra{\phi}^{\otimes t}]
        = \Phi^{(t)}_\mathcal{U}\left(\mathbb{E}_{\phi\sim\mathcal{E}'(\psi)}[\ketbra{\phi}^{\otimes t}]\right).
    \end{equation}
    Due to the unitary invariance of the Haar random, we have
    \begin{equation}
        \rho^{(t)}_\mathrm{Haar} = \Phi^{(t)}_\mathcal{U}\left(U^{\otimes t}\rho^{(t)}_\mathrm{Haar}U^{\otimes t,\dagger}\right).
    \end{equation}
    These imply that
    \begin{equation}
        \mathrm{TD}\left(\mathbb{E}_{\phi\sim\mathcal{E}(\psi)}[\ketbra{\phi}^{\otimes t}],\rho^{(t)}_\mathrm{Haar}\right)=\mathrm{TD}\left(\Phi^{(t)}_\mathcal{U}\left(\mathbb{E}_{\phi\sim\mathcal{E}'(\psi)}[\ketbra{\phi}^{\otimes t}]\right),\Phi_\mathcal{U}^{(t)}\left(\rho^{(t)}_\mathrm{Haar}\right)\right).
    \end{equation}
    Due to the strong convexity of the trace distance~\cite{Nielsen_Chuang_2010}, this is upper bounded by
    \begin{equation}
            \mathrm{TD}\left(\Phi^{(t)}_\mathcal{U}\left(\mathbb{E}_{\phi\sim\mathcal{E}'(\psi)}[\ketbra{\phi}^{\otimes t}]\right),\Phi^{(t)}_\mathcal{U}\left(\rho^{(t)}_\mathrm{Haar}\right)\right)
            \leq\mathrm{TD}\left(\mathbb{E}_{\phi\sim\mathcal{E}'(\psi)}[\ketbra{\phi}^{\otimes t}],\rho^{(t)}_\mathrm{Haar}\right),
    \end{equation}
    where we use the fact that the trace distance satisfies $\mathrm{TD}(\rho,\sigma) = \mathrm{TD}(U\rho U^\dagger,U\sigma U^\dagger)$ for any states $\rho$ and $\sigma$, and any unitary $U$.
    Combining altogether, we have
    \begin{equation}
            \mathrm{TD}\left(\mathbb{E}_{\phi\sim\mathcal{E}(\psi)}[\ketbra{\phi}^{\otimes t}],\rho^{(t)}_\mathrm{Haar}\right)
            \leq \mathrm{TD}\left(\mathbb{E}_{\phi\sim\mathcal{E}'(\psi)}[\ketbra{\phi}^{\otimes t}],\rho^{(t)}_\mathrm{Haar}\right).
    \end{equation}
\end{proof}
Using this, we can immediately deduce the following fact for $\mathcal{E}_\mathrm{ent.}$ and $\mathcal{E}_\mathrm{coh.}$.
\begin{corollary}\label{thm:maximal-randomness}
    For any $\ket{\psi}$, $\mathcal{E}_\mathrm{ent.}(\psi)$ and $\mathcal{E}_\mathrm{coh.}(\psi)$ are closest ensembles to the Haar random states without using entangling and coherent gates, respectively.
\end{corollary}
\begin{proof}
    The set of unitary operators that preserve resource such as entanglement and coherence forms a group under the matrix multiplication. Thus, we can directly use \textbf{Theorem}~\ref{thm:uniform-maximal}.
\end{proof}

\subsection*{6. Tightest-ness of second R\'enyi entropies (proof of Theorem 7)}
We first list some asymptotic formulas for the R\'enyi entropy which will be useful for proving the theorem in this section. First, for some integer $L$ such that $2^n/L$ is an integer, let us set a distribution $p$ in a $L$-dimensional space be $\{1-\delta,\delta/(L-1),\cdots\}$. Here, $\delta$ is set by a constant number $\delta>0$. Later, we will set $\delta$ to depend on $N$. The $\alpha$-R\'enyi entropy is then given by
\begin{equation}
    \mathcal{R}_\alpha(p) = 
    \begin{cases}
        \frac{1}{1-\alpha}\sum_{x\in[L]} p_x^\alpha & \alpha\neq 1\\
        -\sum_{x\in [L]}p_x \log p_x & \alpha = 1
    \end{cases}.
\end{equation}
Let us separately consider the cases when $\alpha<1$ and $\alpha=1$. For $\alpha<1$, we simply have
\begin{equation}\label{eq:Renyi-l1}
    \mathcal{R}_\alpha(p)\sim\frac{\alpha}{1-\alpha}\log\frac{L-1}{\delta}
\end{equation}
for large $L$. For $\alpha=1$, we get
\begin{equation}\label{eq:Renyi-e1}
    \mathcal{R}_1(p) \sim \delta \log L.
\end{equation}
On the other hand, $\mathcal{R}_2(p)$ scales as
\begin{equation}\label{eq:Renyi-2}
    \mathcal{R}_2(p) \sim 2\log\frac{1}{1-\delta}.
\end{equation}
Next, let us set $p'$ by making $1-\delta$ scale as $L^{-1/2}$. Then, for $\alpha>1$, we have
\begin{equation}\label{eq:Renyi-g1}
    \begin{split}
        \mathcal{R}_{\alpha}(p')
        &= \frac{1}{1-\alpha}\log\left( L^{-\alpha/2}+\frac{(1-L^{-1/2})^{-\alpha}}{(L-1)^{\alpha-1}}\right)\\
        &\sim \frac{\alpha}{2(\alpha-1)}\log L.
    \end{split}
\end{equation}
We prove the following theorem using these asymptotic formulas. This theorem shows that only the second R\'enyi entropies, not other $\alpha$-R\'enyi entropies, can give tight bounds on the maximal achievable Haar randomness for arbitrary states. 
\begin{theorem}
    Let $\epsilon$ be a positive number. If the second R\'enyi entropy $\mathcal{R}_2(\psi)$ tightly bounds $\epsilon$ as $\Theta(\exp(-\mathcal{R}_2(\psi)))$, i.e., 
    \begin{equation}
       C_1(t) e^{-\mathcal{R}_2(\psi)} \leq \epsilon \leq C_2(t) e^{-\mathcal{R}_2(\psi)},
    \end{equation}
    for all $\ket{\psi}$ and some positive functions $C_1$ and $C_2$, then there is no $\alpha\neq 2$ that gives 
    \begin{equation}
       C'_1(t) e^{-a\mathcal{R}_\alpha(\psi)} \leq \epsilon \leq C'_2(t) e^{-b\mathcal{R}_\alpha(\psi)}
    \end{equation}
    for all $\ket{\psi}$. Here, $C'_1$ and $C'_2$ are some positive functions, and $a\leq 1$ and $b\geq 1$ are some constants.
\end{theorem}
\begin{proof}
    Suppose there exists $\alpha\neq 2$ such that 
    \begin{equation}
        C'_1(t) e^{-a\mathcal{R}_{\alpha}(\psi)} \leq \epsilon \leq C'_2(t) e^{-b\mathcal{R}_{\alpha}(\psi)}
    \end{equation}
    Then, we have
    \begin{align}
        C'_1(t) e^{-a\mathcal{R}_{\alpha}(\psi)} &\leq C_2(t) e^{-\mathcal{R}_2(\psi)}\\
        C_1(t) e^{-\mathcal{R}_2(\psi)} &\leq C'_2(t) e^{-b\mathcal{R}_{\alpha}(\psi)}.
    \end{align}
    Let us separately consider the cases when $\alpha\leq 1$, $1\leq \alpha < 2$ and $\alpha>2$. When $\alpha \leq 1$, due to Eq.~\eqref{eq:Renyi-l1} and~\eqref{eq:Renyi-e1}, there exists a $n$-qubit state $\ket{\phi_{01}}$ such that
    \begin{equation}
        \mathcal{R}_{\alpha}(\phi_{01}) \sim \log 2^{\beta n},
    \end{equation}
    with $\beta=2^n/L$ and due to Eq.~\eqref{eq:Renyi-2} we have
    \begin{equation}
        \mathcal{R}_2(\phi_{01}) = \Theta(1).
    \end{equation}
    Combining these, for some positive constants $c_{01}$ and $c_2$, we get
    \begin{equation}
        C_1(t) e^{-c_2} \lesssim C'_2(t)2^{-\beta b c_{01} n},
    \end{equation}
    which is false for large enough $n$. Next, when $1\leq \alpha < 2$, due to Eq.~\eqref{eq:Renyi-g1}, there exists a $n$-qubit state $\ket{\phi_{12}}$ such that
    \begin{equation}
        \mathcal{R}_\alpha(\phi_{12}) \sim \frac{\alpha}{2(\alpha-1)}\log 2^{\beta n}
    \end{equation}
    and
    \begin{equation}
        \mathcal{R}_2(\phi_{12}) \sim \log 2^{\beta n}.
    \end{equation}
    These imply that
    \begin{equation}
        C_1(t) 2^{-\beta n} \lesssim C'_2(t)2^{-\frac{\alpha \beta b n}{2(\alpha - 1)}},
    \end{equation}
    which is again false for large enough $n$ as $\alpha/(\alpha-1)>2$ and $b\geq 1$. Finally, let us consider the case with $\alpha > 2$. In this case, for the same state $\ket{\psi_{12}}$, we have
    \begin{equation}
        C'_1(t) 2^{-\frac{\alpha \beta a n}{2(\alpha - 1)}} \lesssim C_2(t) 2^{-\beta n}.
    \end{equation}
    This does not hold for large enough $n$ as $\alpha/(\alpha-1)<2$ and $a\leq 1$. Thus, there is no such $\alpha$.
\end{proof}

\subsection*{7. Achieving near maximal Haar randomness with constant resources (proof of Theorem 8)}
\begin{theorem}\label{thm:rsp}
    For a constant $t\geq 2$, there is a quantum state $\ket{\psi}$ that gives $\mathcal{E}=\mathcal{E}_{\rm ent.}(\psi)\cap\mathcal{E}_{\rm coh.}(\psi)$ forming an approximate state $t$-design with an error that scales with $\Theta(e^{-\mathcal{N}_2(\psi)})$ and $\Theta(e^{-\mathcal{C}_2(\psi)})$. 
\end{theorem}
\begin{proof}
    We consider the following quantum state:
    \begin{equation}
        \ket{\psi} = \frac{1}{\sqrt{K}}\sum_{z\in[K]}\ket{z}_A\otimes\ket{z}_B
    \end{equation}
    for some $k\leq \min(|A|,|B|)$ and $K=2^k$. This state have $\mathcal{N}_2(\psi)=\mathcal{C}_2(\psi)=k\log 2$. We then apply entanglement and coherence preserving unitary operators to $\ket{\psi}$. These operators are given by tensor products of operators in $\mathcal{PF}_A=\{PF|P\in\mathcal{P}(2^{|A|},F\in[U(1)]^{\oplus 2^{|A|}})\}$ and $\mathcal{PF}_B=\{PF|P\in\mathcal{P}(2^{|B|},F\in[U(1)]^{\oplus 2^{|B|}})\}$. We define $\mathcal{E}$ as the ensemble of states generated from $\ket{\psi}$ by applying these unitary operators. A state $\ket{\phi}$ in this ensemble is given by
    \begin{equation}
        \ket{\phi} = \frac{1}{\sqrt{K}}\sum_{z\in[K]}e^{2\pi i(f_A(z)+f_B(z))}\ket{p_A(z)}\otimes\ket{p_B(z)}
    \end{equation}
    for some independent random permutations $p_A:[2^{|A|}]\rightarrow[2^{|A|}]$ and $p_B:[2^{|B|}]\rightarrow[2^{|B|}]$ as well as independent random functions $f_A:[2^{|A|}]\rightarrow [0,1)$ and $f_B:[2^{|B|}]\rightarrow [0,1)$. Then, the $t$-th moments of this ensemble is given by
    \begin{equation}
        \begin{split}
            \mathbb{E}_{\phi\sim\mathcal{E}}[\ketbra{\phi}^{\otimes t}]
            &= \frac{1}{K^t} \sum_{\{z_i\}_{i=1}^t,\{z'_i\}_{i=1}^t \in [K]^t} \mathbb{E}_{p_A,p_B,f_A,f_B} e^{2\pi i(\sum_i f_A(z_i)+f_B(z_i)+f_A(z'_i)+f_B(z'_i))}\\
            &\qquad\qquad\qquad\qquad\qquad\times\ketbra{\{p_A(z_i)\}}{\{p_A(z'_i)\}}\otimes\ketbra{\{p_B(z_i)\}}{\{p_B(z'_i)\}}\\
            &= \frac{1}{K^t} \sum_{\{z_i\}_{i=1}^t \in [K]^t}  \mathbb{E}_{p_A,p_B} \ketbra{\{p_A(z_i)\}}\otimes\ketbra{\{p_B(z_i)\}}\\
            &= \frac{1}{K^t} \sum_{\{z_i\}_{i=1}^t \in [K]^t_\mathrm{dist}} \mathbb{E}_{p_A,p_B}\ketbra{\{p_A(z_i)\}}\otimes\ketbra{\{p_B(z_i)\}}+O(t^2/K)\\
            &= \frac{(2^{|A|}-t)!}{(2^{|A|})!}\frac{(2^{|B|}-t)!}{(2^{|B|})!}\frac{K!}{K^t(K-t)!}\sum_{\substack{\{a_i\}_{i=1}^t \in [2^{|A|}]^t_\mathrm{dist}\\\{b_i\}_{i=1}^t \in [2^{|B|}]^t_\mathrm{dist}}} \ketbra{\{a_i\}}\otimes\ketbra{\{b_i\}}+O(t^2/K)\\
            &= \frac{1}{2^{nt}}\frac{K!}{K^t(K-t)!}\sum_{\{z_i\}_{i=1}^t \in [2^n]^t_\mathrm{dist}} \ketbra{\{z_i\}}+O(t^2/K)\\
            &= \rho^{(t)}_\mathrm{Haar}+O(t^2/K).
        \end{split}
    \end{equation}
    Therefore, $\mathcal{E}$ forms an approximate state $t$-design with error $O(t^2/K)$.
\end{proof}

\bibliography{Ref}